\documentclass[pra,showkeys,showpacs,superscriptaddress,nofootinbib,twocolumn,10pt]{revtex4-1}
\usepackage{graphicx}  
\usepackage{amsthm}
\usepackage{amsmath,amssymb,amsfonts,bbm,MnSymbol,mathrsfs,xfrac}   
\usepackage[breaklinks=true,colorlinks=true,linkcolor=blue,urlcolor=blue,citecolor=blue]{hyperref}
\usepackage{comment}
\usepackage{color}
\usepackage[makeroom]{cancel}
\pdfoutput=1

\renewcommand{\[}{\begin{equation}}
\renewcommand{\]}{\end{equation}}
\renewcommand{\Re}{\mathfrak{Re}}

\newcommand{\ket}[1]{|#1\rangle}
\newcommand{\bra}[1]{\langle#1|}
\newcommand{\braket}[2]{\langle#1|#2\rangle}
\newcommand{\pro}[2]{|#1\rangle\langle#2|}
\newcommand{\mean}[1]{\langle#1\rangle}
\newcommand{\abs}[1]{|#1|}

\newcommand{\tr}{\mathrm{tr}}

\newcommand{\norm}[1]{\left\lvert\left\lvert#1\right\rvert\right\rvert}

\newcommand{\R}{\hat{\rho}}
\renewcommand{\L}{\hat{L}}

\newcommand{\HS}{\mathcal{H}}

\newcommand{\de}{{\mathrm{d}\epsilon}}
\newcommand{\dR}{{\mathrm{d}\R}}

\newcommand{\bu}{{\boldsymbol{u}}}

\newcommand{\be}{{\boldsymbol{\epsilon}}}
\newcommand{\bdeps}{{\mathrm{d}\!\!\;\boldsymbol{\epsilon}}}

\newtheorem{theorem}{Theorem}
\newtheorem{corollary}{Corollary}
\newtheorem{example}{Example}

\begin{document}

\title{Discontinuities of the quantum Fisher information and the Bures metric}

\author{Dominik \v{S}afr\'{a}nek}
\email{dsafrane@ucsc.edu}
\address{SCIPP, University of California, Santa Cruz, Natural Sciences II No.~337, 1156 High Street, Santa Cruz, California 95064, USA}

\date{\today}

\begin{abstract}
We show that two quantities in quantum metrology that were thought to be the same, the quantum Fisher information matrix and the Bures metric, are not the same. They differ at points at which the rank of the density matrix changes. The quantum Fisher information matrix is discontinuous at these points. However, these discontinuities are removable in some sense. We show that the expression given by the Bures metric represents the continuous version of the quantum Fisher information matrix. We also derive an explicit formula for the Bures metric for both singular and non-singular density matrices.
\end{abstract}


\maketitle

\section{Introduction}

The quantum Fisher information and the Bures metric are cornerstones of modern quantum metrology and quantum information geometry. 
They give the ultimate precision bound on the estimation of a parameter encoded in a quantum state known as the Cram\'er-Rao bound\cite{BraunsteinCaves1994a,Paris2009a}. This bound gives the theoretical framework for maximizing sensitivity of new-era quantum detectors such as recently improved~\cite{aasi2013enhanced} gravitational wave detector LIGO that confirmed the last missing piece in the Einstein's theory of relativity~\cite{abbott2016observation}. The quantum Fisher information and the Bures metric have been also used in the description of criticality and quantum phase transitions under the name of `fidelity susceptibility' where they help to describe a sudden change of a quantum state when an external parameter such as temperature is varied~\cite{paraoanu1998bures,zanardi2007bures,venuti2007quantum,gu2010fidelity}. Last but not least, these measures also give the speed limits on the evolution of quantum states~\cite{margolus1998maximum,pires2016generalized}. They have been used for example to estimate the speed limits of quantum computation~\cite{lloyd2000ultimate} or speed limits in charging of batteries~\cite{binder2015quantacell}.

Introduced by Holevo~\cite{Holevo2011a}, Helstrom~\cite{helstrom1967minimum,helstrom1976quantum}, by Bures~\cite{Bures1969a}, and later popularized by Braunstein and Caves~\cite{BraunsteinCaves1994a}, the quantum Fisher information and the Bures metric describe limits in distinguishability of infinitesimally close quantum states $\R_\be$ and $\R_{\be+\bdeps}$ that differ only by a small variation in parameters that parametrize them. To explain, assume we perform a measurement on these two states to distinguish them. We obtain two different statistics of measurement outcomes, and how well we can distinguish between these two statistics is given by a measure known as the Fisher information. Since statistics of measurement outcomes depend on both the quantum states and the chosen measurement, and because some measurements can lead to statistics that are easier to distinguish, to obtain the ultimate precision with what we can distinguish between the  two close states we have to optimize over all such measurements. This then gives rise to the quantum Fisher information, which is a function only of the density matrix $\R_\be$~\cite{Paris2009a}. Since the density matrix depends on the parameter to be estimated, distinguishing between two close density matrices $\R_\be$ and $\R_{\be+\bdeps}$ is equivalent to distinguishing between two close parameters $\be$ and $\be+\bdeps$ that parametrize them. As a result, the quantum Fisher information measures how well the parameter $\be$ itself can be estimated.

Despite the fact that both the quantum Fisher information and the Bures metric have been widely used before, they still contain a large number of strange and unexplored properties. For example, although these quantities are widely believed to be equal, finding the true connection between them is rather elusive. It has been shown that these quantities are the same when two infinitesimally close states that are being compared are pure~\cite{wootters1981statistical}, or when they are both described by a full rank density matrix~\cite{hubner1992explicit,BraunsteinCaves1994a}. It has been suggested that this is also true when the density matrices are of arbitrary rank~\cite{sommers2003bures}; however, we will show that this is not true in general.

The quantum Fisher information and the Bures metric also exhibit mathematical features that make them difficult to handle and that are uncommon in physics. Expressions for the quantum Fisher information or the Bures metric that are valid for a certain quantum state are often undefined  for a state of a lower rank. Formulas for the full-rank or one-rank density matrices are usually easy to obtain but hard to connect to each other even when using appropriate limits, and deriving expressions for density matrices of an arbitrary rank is much harder. To connect known expressions or to derive new ones unintuitive regularization procedures have to be employed~\cite{Monras2013a,Safranek2015b}, or uncommon operators such as the Moore-Penrose pseudoinverse have to be introduced~\cite{Monras2013a}. Due to these difficulties it is common in literature that discussions of glaringly pathological behavior such as $\tfrac{0}{0}$ of the derived expressions are often omitted. It is implicitly assumed that such expressions are either invalid when they are undefined, or that such expressions are still valid and the pathological terms are either set to zero~\cite{Pinel2013b} or to some other value depending on a particular limit involved~\cite{Safranek2015b}. Performing those limits shows that expressions for the quantum Fisher information exhibit strange jumps (discontinuities) when a mixed state approaches a pure state, suggesting that the physics of pure states should differ from the physics of mixed states, a surprising statement on its own.

In this paper we study and describe this strange behavior and expose places where we can expect discontinuities in the figures of merit even for density matrices that are analytical functions of the estimated parameters. We find that these discontinuities happen when a small change in the parameter of the density matrix changes the rank of the density matrix. This is also when two figures of merit of the local estimation theory --- the quantum Fisher information matrix and the Bures metric --- do not coincide. Such scenarios are common when estimating noise in a quantum system~\cite{monras2007optimal,Invernizzi2011a,spedalieri2016thermal} or when the parameter that we want to estimate is encoded into some larger quantum state while an experimentalist has access only to a smaller subsystem~\cite{gambetta2001state,tsang2013quantum,alipour2014quantum}. This is inevitable in a quantum field theory in curved space-time because there are infinitely many modes that need to be traced over~\cite{aspachs2010optimal,dragan2011quantum,Wang2014a,Tian2015a,Safranek2015a,kish2016estimating}. In all of these scenarios the change in the parameter changes the purity of a state. Therefore the rank of the density matrix can also change, which ultimately leads to a discontinuity.

This paper is structured as follows. We first give the necessary background and we review literature published on the topic. Then we present our main results: the relation between the quantum Fisher information matrix and the Bures metric (theorem~\ref{thm:theorem1}, corollary~\ref{cor:equality_conditions_for_QFIs}, and theorem~\ref{thm:theorem2}); continuity of the Bures metric (theorem~\ref{thm:theorem2}); discontinuities of the quantum Fisher information matrix and the Bures metric (corollary~\ref{cor:discontinuous_QFI} and theorem~\ref{thm:discontinuities}); and an expression for the quantum Fisher information matrix of any state as a limit of the quantum Fisher information matrix of a mixed state (theorem~\ref{thm:regularization_procedure}). We accompany our text by three examples and four figures for better understanding. Finally, we discuss possible physical interpretations of points of discontinuity and conclude.

\section{Background}

We use the following notation: We denote the vector of parameters as $\be=(\epsilon_1,...,\epsilon_n)$, and we denote the density matrix dependent on this vector as $\R_{\be}$. If a symbol with an index appears under the sum, the sum goes over all values of the index such that the property is satisfied. For example, $\sum_{p_k>0}$ means that the sum goes over all $k$ such that $p_k>0$. If there is no condition present, the sum goes over all indices written under the sum. We also usually drop writing the explicit dependence on the vector of parameters $\be$ unless we want to stress out this dependence. For example, instead of $p_i(\be)$ we often write only $p_i$, but for $p_i(\be+\bdeps)$ we write the full form. $\bdeps=(\de_1,...,\de_n)$ denotes a small variation in vector $\be$. We denote partial derivatives as $\partial_i\equiv\partial_{\epsilon_i}$ and $\partial_{ij}\equiv\partial_{\epsilon_i}\partial_{\epsilon_j}$. Derivatives with respect to elements of $\bdeps$ will be denoted as $\partial_{\de_i}$ for the first derivatives, and $\partial_{\de_i\de_j}$ for the second derivatives. Elements of a matrix will be denoted by upper indices, e.g.,~$H^{ij}$, while different matrices or operators will be denoted by lower indices, e.g.~$\L_i$. We also write the spectral decomposition of the density matrix as
\[\label{def:spectral_decomposition}
\R_{\be}=\sum_kp_k\pro{k}{k}.
\]

We define symmetric logarithmic derivatives $\L_i$~\cite{Paris2009a} as operator solutions to equations
\[\label{def:SLDs}
\frac{1}{2}\big(\L_{i}\R_{\be}+\R_{\be} \L_{i}\big)=\partial_{i}\R_{\be}.
\]
The quantum Fisher information matrix is then a symmetric positive or a positive semi-definite matrix defined as~\cite{Paris2009a}
\[\label{def:Information_matrix}
H^{ij}(\boldsymbol{\epsilon}):=\frac{1}{2}\tr\left[(\L_{i}\L_{j}+\L_{j}\L_{i})\R_{\be}\right].
\]
Using the spectral decomposition~\eqref{def:spectral_decomposition} of the density matrix it is relatively easy\footnote{Assuming $\R_{\be}\in C^{(1)}$, inserting $\L_i$ to Eq.~\eqref{def:SLDs} gives the left hand side $LHS=\sum_{p_k+p_l>0}\bra{k}\partial_{i}\R_{\be}\ket{l}\pro{k}{l}$ which together with $\sum_{p_k=p_l=0}\bra{k}\partial_{i}\R_{\be}\ket{l}\pro{k}{l}=0$ gives the right hand side. The second identity comes from the fact that for $\be$ such that $p_k(\be)=0$, also $\partial_i p_k(\be)=0$ because $p_k$ reaches the local minimum at point $\be$.} to check that Eqs.~\eqref{def:SLDs} have solutions ${\L_i=2\sum_{p_k+p_l> 0}\frac{\langle k|\partial_i\R_{\epsilon}|l\rangle}{p_k+p_l} |k\rangle\langle l|}$. Inserting these expressions into Eq.~\eqref{def:Information_matrix} gives the quantum Fisher information matrix,
\begin{equation}\label{QFI_multi}
H^{ij}(\be)=2\!\!\!\!\sum_{p_k+p_l> 0}\!\!\!\!\frac{\Re(\langle k|\partial_i\R_{\be}|l\rangle\langle l|\partial_j\R_{\be}|k\rangle)}{p_k+p_l},
\end{equation}
where $\Re$ denotes the real part.

The quantum Fisher information matrix is the figure of merit in the multi-parameter quantum Cram\'er-Rao bound which gives a lower bound on the covariance matrix of the vector of locally unbiased estimators $\hat{\be}$~\cite{szczykulska2016multi,Paris2009a},
\[
\mathrm{Cov}[\hat{\boldsymbol{\epsilon}}]\geq{H}^{-1}(\boldsymbol{\epsilon}).
\]
$\mathrm{Cov}[\hat{\boldsymbol{\epsilon}}]=\mean{\hat{\epsilon}_i\hat{\epsilon}_j}-\mean{\hat{\epsilon}_i}\mean{\hat{\epsilon}_j}$ is the covariance matrix and $H^{-1}(\boldsymbol{\epsilon})$ the inverse of the matrix defined in Eq.~\eqref{def:Information_matrix}. The above equation should be understood as an operator inequality. It states that $\mathrm{Cov}[\hat{\boldsymbol{\epsilon}}]-{H}^{-1}$ is a positive semi-definite or a positive definite matrix. 

In this paper we study the connection between the quantum Fisher information matrix and the Bures metric (also known the ``statistical distance'' in older literature~\cite{wootters1981statistical,BraunsteinCaves1994a} and as the ``fidelity susceptibility'' in the condensed matter theory literature~\cite{venuti2007quantum}). To define the Bures metric~\cite{Bures1969a} we first introduce the Bures distance. The Bures distance is a measure of distinguishability between two quantum states $\R_{1,2}$ and it is defined through the Uhlmann fidelity~\cite{Uhlmann1976a}
\[\label{def:Uhlmann_Fidelity}
\mathcal{F}({\R}_{1},{\R}_{2})\,:=\,\Big(\tr\sqrt{\sqrt{{\R}_{1}}\,{\R}_{2}\,\sqrt{{\R}_{1}}}\Big)^{2}
\]
as
\[\label{def:bures_distance}
d_B^2(\R_1,\R_2)=2\big(1-\sqrt{\mathcal{F}(\R_1,\R_2)}\big).
\]
The Bures distance gives rise to the Bures metric $g^{ij}$ through the definition for the line element,
\[\label{eqn:bures}
\sum_{i,j}g^{ij}(\be)\mathrm{d}\epsilon_i\mathrm{d}\epsilon_j:=d_B^2(\R_{\be},\R_{\be+\bdeps}).
\]
This definition shows that the Bures metric measures the amount of distinguishability between two close density matrices $\R_{\be}$ and $\R_{\be+\bdeps}$ in the coordinate system $\be$. Precisely speaking, the above equation defines a metric tensor (or simply metric) $g^{ij}$ induced by the Bures distance. The coordinate system $\be$ is not required to describe the entire manifold of density matrices but can rather define a submanifold. The metric~\eqref{eqn:bures} is then an induced metric on this submanifold. We define quantity $H_c$ (which we will later call the continuous quantum Fisher information matrix) as four times the Bures metric,
\[\label{def:continuous_QFI}
H_c:=4g.
\]

The connection between the quantum Fisher information matrix and the Bures metric was extensively studied in literature, particularly in several papers deriving explicit formulas for the statistical distance, the Bures metric, or the infinitesimal Bures distance~\cite{Bures1969a,helstrom1976quantum,Holevo2011a,Uhlmann1976a,wootters1981statistical,hubner1992explicit,hubner1993computation,BraunsteinCaves1994a,Jozsa1994a,slater1996quantum,slater1998bures,dittmann1999explicit,sommers2003bures,zanardi2007bures,Paris2009a}. We also point out papers related to the continuity of the quantum Fisher information~\cite{wang2015discontinuity,rezakhani2015continuity}.

There are three papers directly related to our study. In the first paper~\cite{BraunsteinCaves1994a} Braunstein and Caves generalized the notion of the statistical distance from pure states to mixed states by maximizing the Fisher information over all possible quantum measurements. In today's terms, the resulting statistical distance is an equivalent of the quantum Fisher information defined in Eq.~\eqref{QFI_multi}. Moreover, it was noted in~\cite{BraunsteinCaves1994a} that the derived expression of the statistical distance is proportional to the infinitesimal Bures distance that was explicitly computed by H\"ubner~\cite{hubner1992explicit}. What was not mentioned, however, is that the results of paper~\cite{hubner1992explicit} are applicable only to non-singular density matrices.\footnote{H\"ubner states ``We assume $A(0)=\rho$ invertible.''} It is discussed there that in the case when $\R$ becomes singular the metric can be regularized by switching to a new set of coordinates and that the metric tensor $g^{ij}$ remains finite. On the other hand, paper~\cite{hubner1992explicit} does not provide an explicit expression for the infinitesimal Bures distance for the case of singular density matrices. Sommers and \.{Z}yczkowski went a bit further by considering also singular density matrices~\cite{sommers2003bures}. However the entire discussion of this topic is reduced to one sentence\footnote{The sentence being ``Note, that if $\rho_{\nu}=0$ and $\rho_{\mu}=0$, $\delta\rho_{\nu\mu}$ does not appear and therefore terms where the denominator vanishes have to be excluded.''} which leads to incorrect conclusions.

Performing the proof in Ref.~\cite{sommers2003bures} in detail using the same argumentation reveals that the resulting expression differs from the one that was published. The infinitesimal Bures distance is actually $d_B^2(\R,\R+\dR)=\frac{1}{2}\sum_{p_k>0,p_l> 0}\frac{|\langle k|\dR|l\rangle|^2}{p_k+p_l}$ and not $d_B^2(\R,\R+\dR)=\frac{1}{2}\sum_{p_k+p_l> 0}\frac{|\langle k|\dR|l\rangle|^2}{p_k+p_l}$ as stated in the paper. Intuitively, the reason why extra terms given by $p_k>0,p_l=0$ and $p_k=0,p_l>0$ do not appear in the sum can be understood in the following way: the argument $d_B^2(\R,\R+\dR)$ takes into account only the first-order correction $\R+\dR$, while the right-hand side depends on the second order $\dR^2$. The extra terms that are missing come from the second order correction to the argument. However, to obtain these extra terms it is necessary to consider the expression given by  $d_B^2(\R_{\be},\R_{\be+\bdeps})$ from Eq.~\eqref{eqn:bures} instead of $d_B^2(\R,\R+\dR)$, as has been done in Ref.~\cite{sommers2003bures} (see Appendix~\ref{app} for more detail).

\begin{figure}[t!]
  \centering
\includegraphics[width=0.95\linewidth]{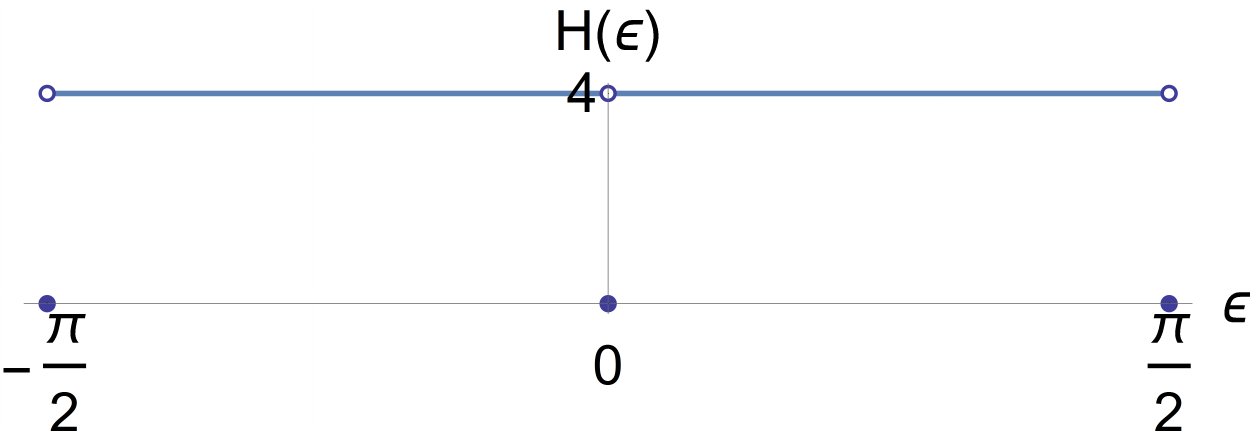}\\
\vspace{0.5cm}
  \includegraphics[width=0.95\linewidth]{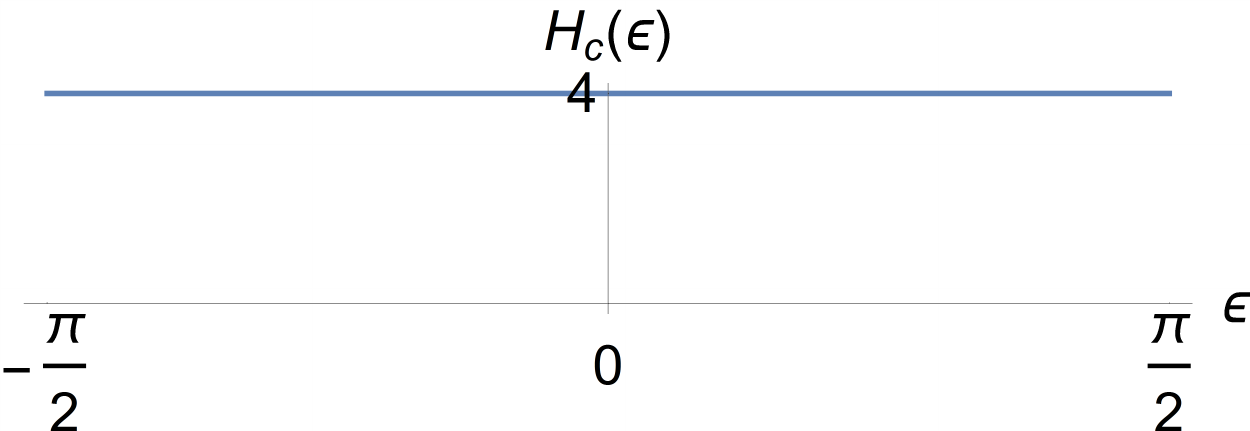}
  \caption{The quantum Fisher information $H$ and the (four times) Bures metric $H_c$ give different results for the same density matrix~\eqref{eq:discontinuous_example0}.}\label{fig:discrepancy}
\end{figure}

We will show in detail that even when considering the right figure of merit $d_B^2(\R_{\be},\R_{\be+\bdeps})$ for the Bures metric, there is still a discrepancy between the Bures metric and the quantum Fisher information matrix in certain cases. To motivate this paper we illustrate this discrepancy in the following example.

\begin{example}\label{ex:example1}
Consider a state where the parameter we estimate characterizes the purity of a quantum state,
\[\label{eq:discontinuous_example0}
\R_{\epsilon}=\sin^2\!\!\!\; \epsilon\,\pro{0}{0}+\cos^2\!\!\!\; \epsilon\,\pro{1}{1}.
\]
The fidelity between two close states can be easily calculated as $\sqrt{\mathcal{F}(\R_{\epsilon},\R_{\epsilon+\de})}=|\sin(\epsilon)\sin(\epsilon\!+\!\de)|+|\cos(\epsilon)\cos(\epsilon\!+\!\de)|$, which inserting into Eq.~\eqref{eqn:bures} and using definition~\eqref{def:continuous_QFI} gives a constant function
\[
H_c(\epsilon)=4.
\]
Using Eq.~\eqref{QFI_multi}, for $\epsilon\neq k\frac{\pi}{2}$, $k\in \mathbb{Z}$ we find $H(\epsilon)=4$. For $\epsilon= k\frac{\pi}{2}$ one term in the sum has its denominator equal to zero ($p_k+p_l=0$) and therefore it is not considered in the sum, while other terms are zero. Together we have
\[\label{eq:disc_H_first_ex}
H(\epsilon)=\begin{cases}
      4 & \epsilon\neq k\frac{\pi}{2} \\
      0 & \epsilon= k\frac{\pi}{2}.
   \end{cases}
\]
Graphs of functions $H$ and $H_c$ are shown in Fig.~\ref{fig:discrepancy}.
\end{example}

This example shows that although the expression given by the quantum Fisher information matrix $H$ and the four times Bures metric $H_c$ give the same results everywhere where the density matrix is full-rank (non-singular), the expressions differ at points $\epsilon$ at which an eigenvalue vanishes. As we will show in the following section, this is a completely general behavior. For parameterized quantum states $\R_{\be}$ in which a slight change in the parameter $\be$ results in an eigenvalue of the density matrix to vanish (or equivalently, results in an eigenvalue to ``pop out''), the two figures of merit do not coincide. This discrepancy will be then connected to the discontinuous behavior of the quantum Fisher information matrix.

\section{Results}

We assume $\R_{\be}\in C^{(2)}$ in all following theorems, i.e., we assume that the second derivative of the density matrix exists and that it is a continuous function. Although the first theorem could be easily modified to require only the existence of the second derivative (with its discontinuity possibly resulting in $\partial_{ij}p_k\neq\partial_{ji}p_k$), the continuity of the second derivative is crucial for other theorems that speak about continuity of the quantum Fisher information matrix and the Bures metric.

\begin{theorem}\label{thm:theorem1}
The Bures metric 
is connected to the quantum Fisher information matrix through the relation
\[\label{eq:bures_metric}
H_c^{ij}(\be)
=H^{ij}(\be)+2\!\!\!\!\sum_{p_k(\be)=0}\!\!\!\!\partial_{ij}p_k(\be).
\]
$p_k(\be)=0$ denotes that the sum goes over all values $k$ such that their respective eigenvalue $p_k$ vanishes at point $\be$. Defining the Hessian matrices as $\mathcal{H}_k^{ij}:=\partial_{ij}p_k$ we can also write Eq.~\eqref{eq:bures_metric} in an elegant matrix form, ${H_c=H+2\sum_{p_k=0}\mathcal{H}_k}$.
\end{theorem}
\begin{proof}
  See Appendix~\ref{app:theorem1}.
\end{proof}

Theorem~\ref{thm:theorem1} gives an explicit formula for the Bures metric for both singular and non-singular density matrices, and it generalizes the result of~\cite{hubner1992explicit} by including singular matrices. It also shows that the (four times) Bures metric and the quantum Fisher information matrix do not coincide only at certain points $\be$ at which an eigenvalue vanishes. When the density matrix $\R_{\be}$ is full rank (non-singular), or when the change of the parameter does not result in the change of purity, for example, when the operation encoding $\be$ is a unitary operation, the (four times) Bures metric $H_c$ and the quantum Fisher information matrix $H$ are identical. It is worth noting that the Hessian matrix $\mathcal{H}_k(\be)$ is a positive or a positive semi-definite matrix because $p_k$ reaches the local minimum at point $\be$ for which $p_k(\be)=0$. We can sum these findings in the following matrix inequality.
\begin{corollary}\label{cor:equality_conditions_for_QFIs}
\[
H_c\geq H,
\]
and $H_c= H$ if and only if for all  $k$ and $\be$ such that $p_k(\be)=0$, $\mathcal{H}_k(\be)=0$.
\end{corollary}
\begin{proof}
The inequality comes from the fact that $\mathcal{H}_k(\be)$ is a positive semi-definite or a positive definite matrix. The equality condition comes directly from Eq.~\eqref{eq:bures_metric}.
\end{proof}

Next we show that the quantity $H_c$ is in a certain sense a continuous version of the quantum Fisher information matrix. The discontinuous points of the quantum Fisher information matrix are redefined as the limits of the quantum Fisher information matrix of nearby points.
\begin{theorem}\label{thm:theorem2}
We denote a unit vector with number $1$ at the $l$'th position as $\boldsymbol{e}_l=(0,\dots,0,1,0\dots,0)$. Then
\[\label{eq:construction_of_continuous_QFI}
H_c^{ij}(\be)=\lim_{\de\rightarrow 0}H^{ij}(\be+\de\,\boldsymbol{e}_i)=\lim_{\de\rightarrow 0}H^{ij}(\be+\de\,\boldsymbol{e}_j).
\]
Moreover, $H_c^{ij}$ is a continuous function in parameter $\epsilon_i$, $\epsilon_j$ respectively, for any fixed parameters $\epsilon_k$ such that $k\neq i$, $k\neq j$ respectively.\footnote{In the proof of this continuity property we also assume that the number of eigenvalues is finite, which leads to $\sum_{p_k=0}\mathcal{O}(\de)=\mathcal{O}(\de)$. This assumption might be problematic, for example, when estimating parameters encoded in Gaussian quantum states because such states live in an infinite-dimensional Hilbert space. However, we believe that this assumption might not be necessary, and it should be possible to derive $\sum_{p_k=0}\mathcal{O}(\de)=\mathcal{O}(\de)$ by showing that the sum converges sufficiently fast. This would nevertheless come from very careful considerations and it could require expanding relevant quantities up to the third order in $\de$.}
\end{theorem}
\begin{proof}
  See Appendix~\ref{app:theorem2}.
\end{proof}

The next corollary will show that the quantum Fisher information matrix is not in general a continuous function even for density matrices that are analytical functions of its parameters.
\begin{corollary}\label{cor:discontinuous_QFI}
If $\sum_{p_k(\be)=0}\partial_{ij}p_k(\be)\neq 0$ then the element of the quantum Fisher information matrix $H^{ij}$ is not continuous at point $\be$.
\end{corollary}
\begin{proof}
Combining theorem~\ref{thm:theorem1} and theorem~\ref{thm:theorem2}, we find
\[
\lim_{\de\rightarrow 0}H^{ij}(\be+\de\,\boldsymbol{e}_i)-H^{ij}(\be)=2\!\!\!\!\sum_{p_k(\be)=0}\!\!\!\!\partial_{ij}p_k(\be)\neq 0,
\]
which by definition means that $H^{ij}$ is not continuous in $\epsilon_i$ at point $\be$ and thus neither is it a continuous function at point $\be$.
\end{proof}

\begin{figure}[t!]
  \centering
\includegraphics[width=1\linewidth]{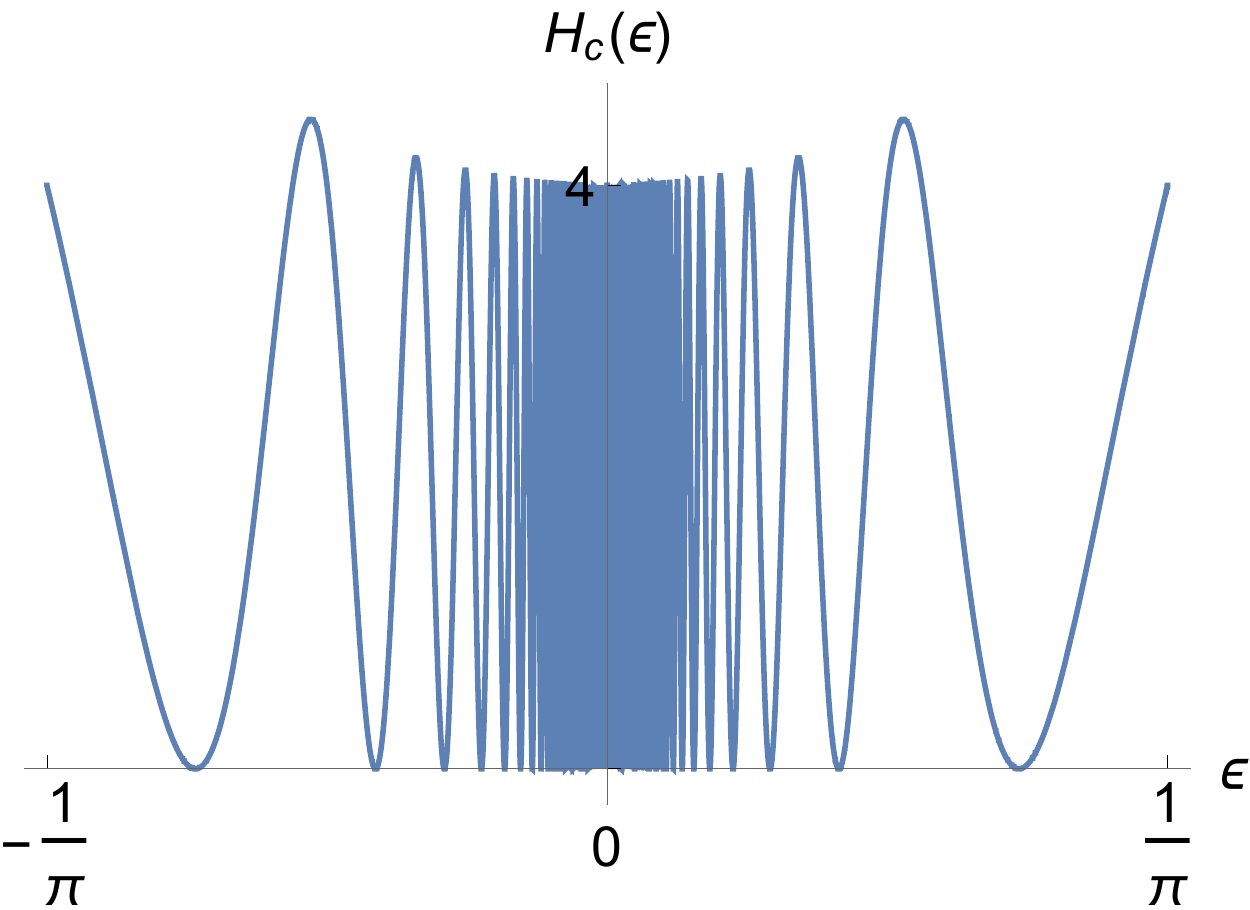}
  \caption{The continuous quantum Fisher information of density matrix
  $\R_{\epsilon}=\begin{cases}
      \epsilon^4\sin^2\!\!\!\; \frac{1}{\epsilon}\,\pro{0}{0}+(1-\epsilon^4\sin^2\!\!\!\; \frac{1}{\epsilon}\,)\pro{1}{1} & \epsilon\neq 0\\
      \pro{1}{1} & \epsilon=0.
   \end{cases}
$
  The second derivative $\partial_{\epsilon\epsilon}\R_{\epsilon}$ exists everywhere but it is discontinuous at point $\epsilon=0$. Theorem~\ref{thm:theorem2} does not apply anymore and $H_c$ does not have to be continuous. An explicit computation shows $H_c(\epsilon)=\begin{cases}
      \frac{4(2\epsilon\sin\!\!\!\; \frac{1}{\epsilon}-\cos\!\!\!\; \frac{1}{\epsilon})^2}{1-\epsilon^4\sin^2\!\!\!\; \frac{1}{\epsilon}} & \epsilon\neq 0\\
      0 & \epsilon=0.
   \end{cases}
$}\label{fig:extreme example}
\end{figure}

Theorem~\ref{thm:theorem2} says when a single parameter $\epsilon$ is estimated, $H_c$ is a continuous function in this parameter. For that reason we call $H_c$ \emph{the continuous quantum Fisher information matrix}. It is important to point out that the assumption required in all of our theorems, $\R_{\be}\in C^{(2)}$, is crucial for theorem~\ref{thm:theorem2} to hold (see Fig.~\ref{fig:extreme example}).  Also, similarly to the quantum Fisher information matrix, the continuous quantum Fisher information matrix is not in general continuous in the topology of multiple parameters $\be=(\epsilon_1,...,\epsilon_n)$. This is precisely stated in the following theorem.
\begin{theorem}\label{thm:discontinuities}
If there exists a unit vector $\boldsymbol{u}=(u_1,\dots,u_n)$ such that
\[\label{eq:Delta}
\Delta_{\boldsymbol{u}}^{ij}(\be)\!:=\!2\!\!\!\!\!\!\!\!\!\!\!\!\!\!\!\!\sum_{\substack{p_{k}= 0, \\\sum_{s,t}\partial_{st}p_{k}u_su_t> 0}}
\!\!\!\!\!\!\!\!\!\!\!\!\!\!\!\bigg(\!\frac{\big(\sum_{m}\partial_{im}p_k u_m\big) \big(\sum_{n}\partial_{jn}p_ku_n\big)}{\sum_{s,t}\partial_{st}p_{k}u_su_t}-\partial_{ij}p_k\!\bigg)\!\neq\! 0,
\]
then the element of the continuous quantum Fisher information matrix $H_c^{ij}$ is not continuous at point $\be$. We can also write Eq.~\eqref{eq:Delta} in an elegant matrix form, ${\Delta_{\boldsymbol{u}}^{ij}(\be):=2\!\sum_{p_{k}= 0,\boldsymbol{u}^T\mathcal{H}_k \boldsymbol{u}>0}\Big(\frac{(\mathcal{H}_k \boldsymbol{u})_i(\mathcal{H}_k \boldsymbol{u})_j}{\boldsymbol{u}^T\mathcal{H}_k \boldsymbol{u}}-\mathcal{H}_k^{ij}\Big)\neq 0}$. 

We denote a unit vector with number $1$ at the $l$'th position as $\boldsymbol{e}_l=(0,\dots,0,1,0\dots,0)$. $H_c^{ij}$ is continuous in $\epsilon_l$ at point $\be$ if and only if $\Delta_{\boldsymbol{e}_l}^{ij}(\be)=0$.

$\Delta_{\boldsymbol{u}}^{ij}$ measures the jump of function $H_c^{ij}$ at point $\be$ when coming from direction $\boldsymbol{u}$,
\[\label{eq:meaning_of_Delta}
\Delta_{\boldsymbol{u}}^{ij}(\be)=\lim_{\de\rightarrow 0}H_c^{ij}(\be+\de\,\bu)-H_c^{ij}(\be).
\]
\end{theorem}
\begin{proof}
  See Appendix~\ref{app:discontinuities}.
\end{proof}

\begin{figure}[t!]
  \centering
\includegraphics[width=0.75\linewidth]{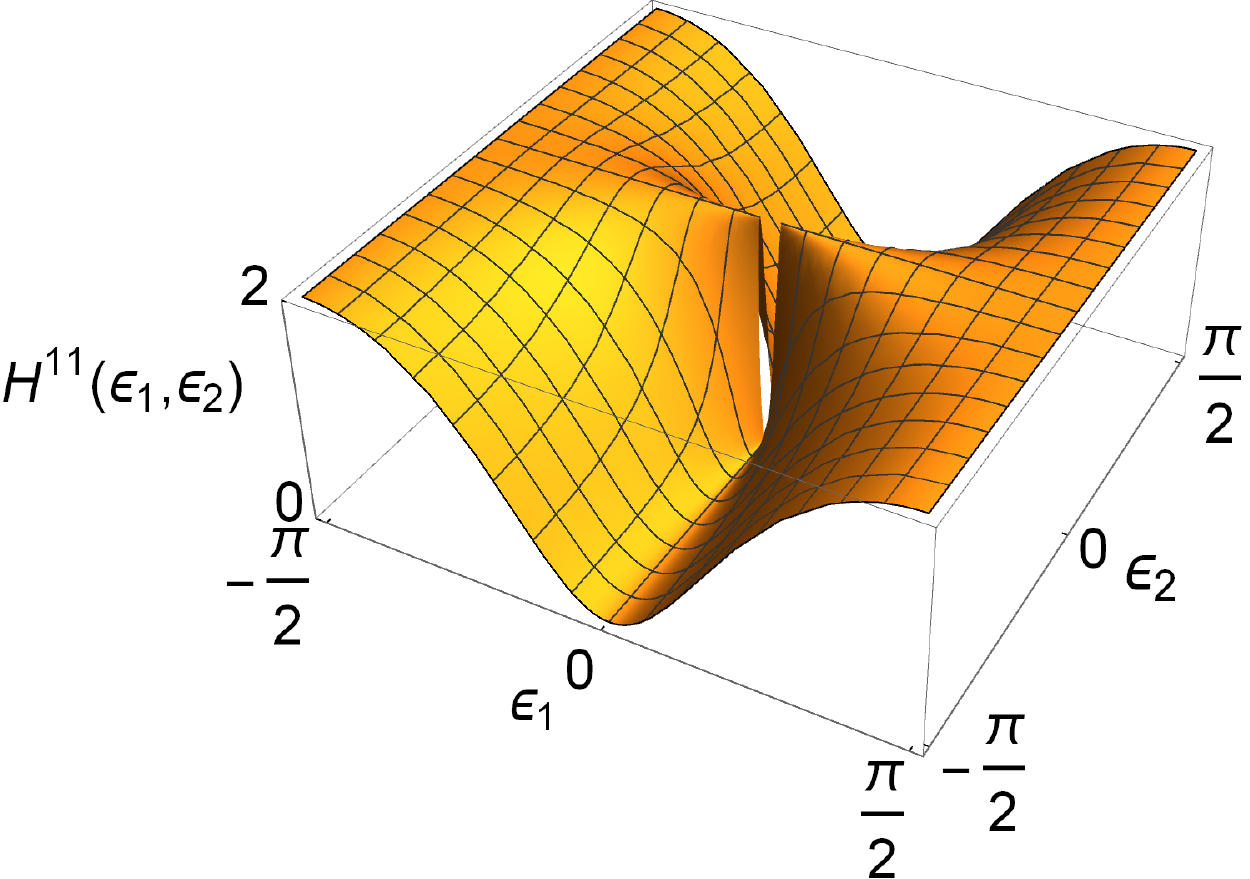}\\
  \includegraphics[width=0.75\linewidth]{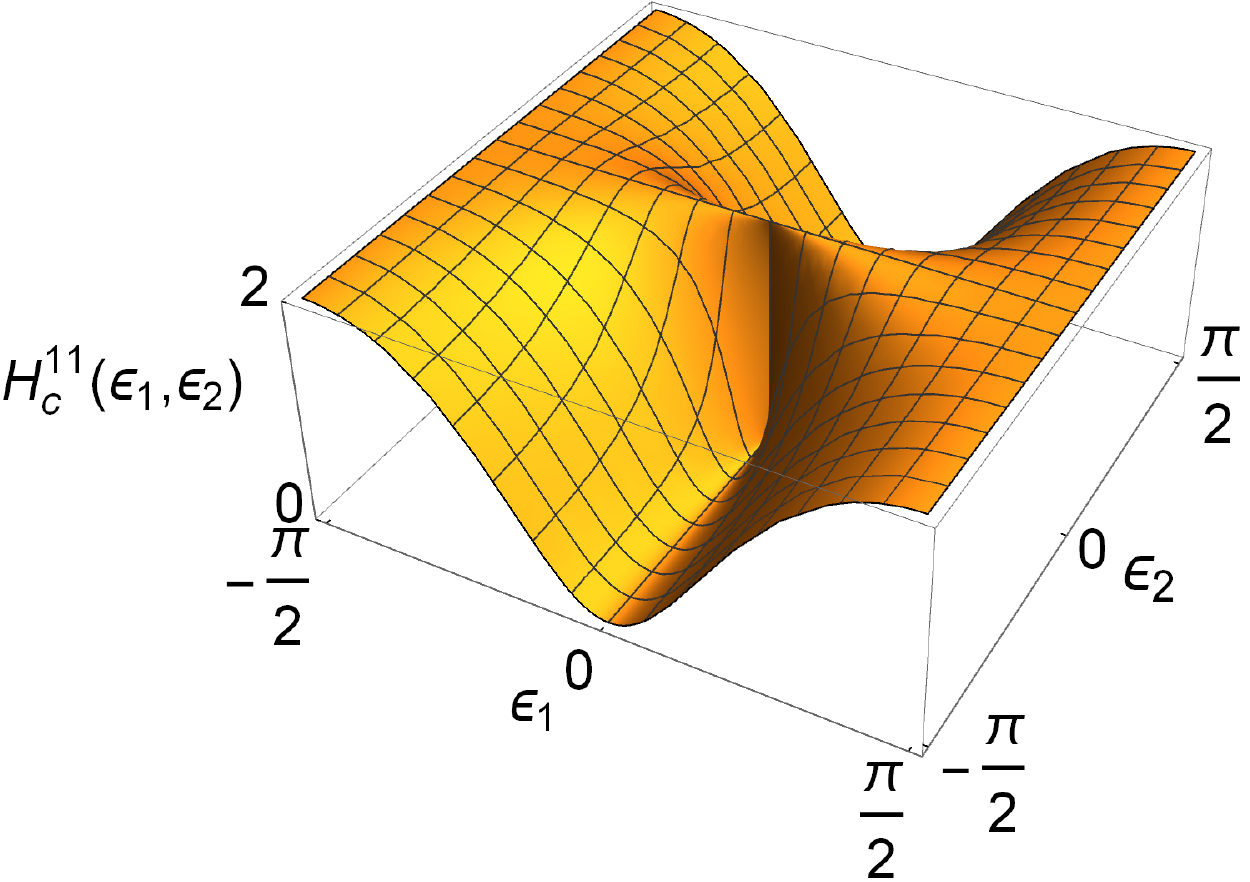}
  \caption{Graphs of the first element of the quantum Fisher information matrix $H^{11}$ and the first element of the continuous quantum Fisher information matrix $H_c^{11}$, for the estimation of parameters of the density matrix~\eqref{eq:discontinuous_example}. These graphs are identical everywhere apart from point $(\epsilon_1,\epsilon_2)=(0,0)$. Clearly, neither function is a continuous function in both parameters at the same time, however, $H_c^{11}$ is guaranteed to be a continuous function in $\epsilon_1$ for any $\epsilon_2$.}\label{fig:Continuous_QFI}
\end{figure}

We illustrate the discontinuous behavior of the quantum Fisher information matrix on the following example.
\begin{example}
Consider a quantum state depending on two parameters,
\[\label{eq:discontinuous_example}
\R_{\be}=\frac{1}{2}(\sin^2\!\!\!\; \epsilon_1+\sin^2\!\!\!\;  \epsilon_2)\pro{0}{0}+\frac{1}{2}\cos^2 \!\!\!\; \epsilon_1\pro{1}{1}+\frac{1}{2}\cos^2\!\!\!\;  \epsilon_2\pro{2}{2}.
\]
We are going to study the first element of the quantum Fisher information matrix $H^{11}$ which measures the mean squared error in the estimating parameter $\epsilon_1$. While the expression for the quantum Fisher information matrix~\eqref{QFI_multi} assigns value $H^{11}(0,0)=0$ to the problematic point $\be=(0,0)$, the continuous quantum Fisher information matrix assigns value $H_c^{11}(0,0)=2$. According to theorem~\ref{thm:theorem2} this definition of the problematic point makes the function $H_c^{11}$ a continuous function in $\epsilon_1$, but not necessarily in $\epsilon_2$.

Using theorem~\ref{thm:discontinuities}, we are going to prove that $H_c$ is not continuous in $\epsilon_2$  at point $\be=(0,0)$. To do that we will study $\Delta_\bu^{11}(0,0)$ from Eq.~\eqref{eq:Delta}. The only relevant eigenvalue is the first one because $p_1(0,0)=0$, while for others $p_2(0,0)=p_3(0,0)=\frac{1}{2}$. The respective Hessian matrix is
\[
\mathcal{H}_1(0,0)=
\begin{bmatrix}
  1 & 0 \\
  0 & 1
\end{bmatrix}
\]
that gives
\[
\Delta_{\boldsymbol{u}}^{11}(0,0)=2\bigg(\frac{u_1^2}{u_1^2+u_2^2}-1\bigg).
\]
Setting $\bu=\boldsymbol{e}_2$ we have $\Delta_{\boldsymbol{e}_2}^{11}(0,0)=-2$, which according to theorem~\ref{thm:discontinuities} means that $H_c^{11}$ is not continuous in $\epsilon_2$ at point $(0,0)$ and thus neither is it a continuous function of $\be$. Graphs of $H^{11}$ and $H_c^{11}$ are shown in Fig.~\ref{fig:Continuous_QFI}.
\end{example}

Theorem~\ref{thm:discontinuities} also states that it is not generally possible to use the multi-parameter Taylor's expansion of the quantum Fisher information matrix at problematic points $\be$ at which $p_k(\be)=0$ and $\mathcal{H}_k(\be)>0$. These are often exactly the points around which we would like to do this expansion, for example when considering a state with a slight impurity that is almost pure. For example, consider a task of estimating both phase and temperature of a quantum state, $\be=(\theta,T)$. Theorem~\ref{thm:discontinuities} says it is not possible in general to approximate the quantum Fisher information matrix by expanding this function in both the phase parameter $\theta$ and the small temperature parameter $T$ at the same time. The directional Taylor's expansion, for example, in parameter $\frac{\theta}{T}$, has to be employed instead. An example that utilizes a variation of this method in a quantum field theory in curved space-time can be found in Ref.~\cite{Safranek2015a}. We will return to this in a moment.

We mentioned in the Introduction that it is often necessary to use different expressions when calculating the quantum Fisher information matrix for density matrices of different ranks. However, we will design a method which will require only one expression to obtain every other expression by performing a certain limit. We are going to show that the quantum Fisher information matrix of any state can be calculated as a limit of the quantum Fisher information matrix of a full-rank state. We call this process \emph{the regularization procedure} in analogy with the result for Gaussian states~\cite{Safranek2015b}. 
\begin{theorem}\label{thm:regularization_procedure}
We define a density matrix
\[\label{eq:reg_example}
\R_{\be,\nu}:=(1-\nu)\R_{\be}+\nu\R_{0},
\]
where $0<\nu<1$ is a real parameter and $\R_{0}$ is any $\be$-independent full-rank density matrix that is diagonal in the eigenbasis of the density matrix $\R_{\be}$. Then the resulting matrix $\R_{\be,\nu}$ is a full-rank matrix and
\begin{align}\label{eq:regularization_procedure_density_matrix}
H(\be)&=\lim_{\nu\rightarrow 0}H(\R_{\be,\nu}),\\
H_c(\be)&=\lim_{\nu\rightarrow 0}H(\R_{\be,\nu})+2\!\!\!\!\!\sum_{p_{k}(\be)= 0}\!\!\!\!\!\mathcal{H}_k(\be).
\end{align}
In finite-dimensional Hilbert spaces $\R_{0}$ can be defined as a multiple of identity, $\R_0=\frac{1}{\mathrm{dim}{\HS}}\hat{I}$.
\end{theorem}
\begin{proof}
$\R_{\be,\nu}$ has eigenvalues equal to $(1-\nu)p_k+\nu p_{0k}>0$, where $p_{0k}>0$ are eigenvalues of $\R_{0}$. Therefore the density matrix is full-rank and the sum in Eq.~\eqref{QFI_multi} for evaluating $H(\R_{\be,\nu})$ goes over all terms.  We evaluate this sum while inserting $\partial_i\R_{\be,\nu}=(1-\nu)\partial_i\R_{\be}$ and perform the limit $\nu\rightarrow 0$. Limits of terms for which $p_k+p_l=0$ are zero and limits of terms for which $p_k+p_l>0$ are identical to terms in the sum for $H(\R_{\be})\equiv H(\be)$. Therefore, performing such a limit gives exactly the quantum Fisher information matrix $H(\epsilon)$. The rest of the statement follows directly from theorem~\ref{thm:theorem1}.
\end{proof}

Theorem~\ref{thm:regularization_procedure} shows that the quantum Fisher information matrix of a pure state can be calculated as a limit of the quantum Fisher information matrix of mixed states, while obtaining all of its possible discontinuities. These discontinuities reflect the results of corollary~\ref{cor:discontinuous_QFI} and results of theorem~\ref{thm:discontinuities}, with $\nu$ acting as an additional parameter with the only difference that now $\nu$ is not the parameter we estimate. It is, for example, possible to use Eq.~\eqref{eq:meaning_of_Delta} to study directional limits in the extended unit vector $\tilde{\boldsymbol{u}}=(u_1,\dots,u_n,u_{\nu})$, where $u_{\nu}$ denotes the amount of direction in the mixedness parameter $\nu$. Such expressions can serve as a good substitute for a possibly non-existent multi-parameter Taylor's expansion: defining $\tilde{\be}=\tilde{\be}_0+\de\,\tilde{\boldsymbol{u}}$, where $\tilde{\be}_0=(\epsilon_1^{(0)},\dots,\epsilon_n^{(0)},\nu^{(0)})$ is the point around which we expand, according to Eq.~\eqref{eq:meaning_of_Delta} the quantum Fisher information matrix can be approximated to the zeroth order by
\[\label{eq:approx_H}
H_c^{ij}(\tilde{\be})\approx H_c^{ij}(\tilde{\be}_0)+\Delta_{\tilde{\boldsymbol{u}}}^{ij}(\be_0),
\]
where $i,j=1,\dots,n$. Of course, other parameters can be used instead of $\nu$, such as the previously mentioned temperature parameter $T$ using the extended unit vector $\tilde{\boldsymbol{u}}=(u_1,\dots,u_n,u_{T})$. The above equation is then what we could call the zeroth order of the directional Taylor's expansion. In case the function $H_c$ is discontinuous at point $\be_0$, the value of $\Delta_{\tilde{\boldsymbol{u}}}^{ij}(\be_0)$ is non-zero and this value measures the amount of jump in function $H_c$ in the direction $\bu$ from the point $\be_0$.

Both the regularization procedure and the directional Taylor's expansion can be demonstrated on the following example. We use the state from the first example, Eq.~\eqref{eq:discontinuous_example0}, and study the quantum Fisher information of the respective regularized state $\R_{\epsilon,\nu}$.

\begin{figure}[t!]
  \centering
\includegraphics[width=0.95\linewidth]{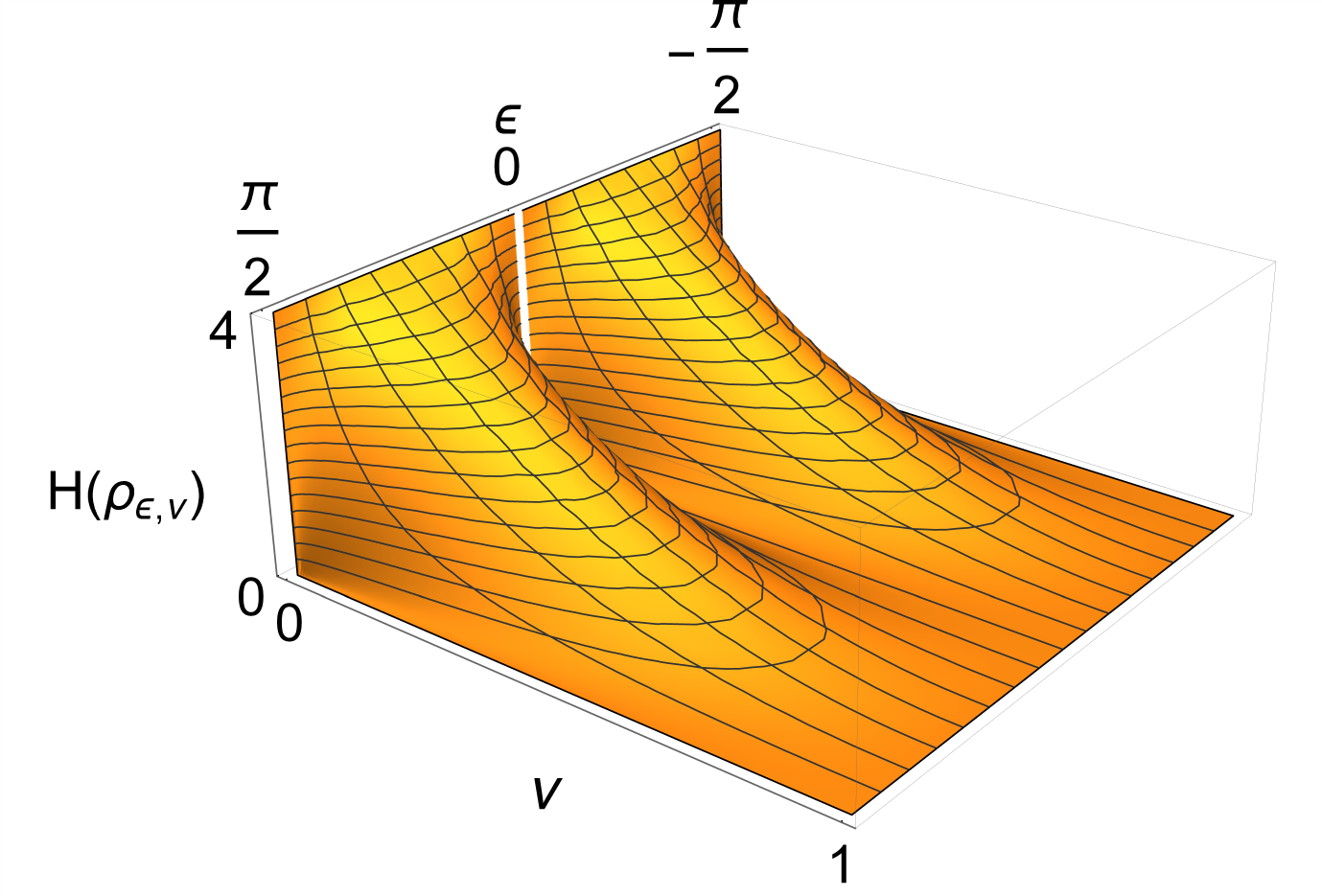}
  \caption{The quantum Fisher information of the regularized state, Eq.~\eqref{eq:discontinuous_example_mixed_H}. Clearly, as $\nu$ approaches zero the function approaches the (discontinuous) quantum Fisher information~\eqref{eq:disc_H_first_ex}, as shown in Fig.~\ref{fig:discrepancy}.}\label{fig:reg}
\end{figure}

\begin{example}
Consider a density matrix dependent on parameter $\epsilon$ and the mixedness parameter $\nu$,
\[\label{eq:discontinuous_example_mixed}
\R_{\epsilon,\nu}=(1-\nu)\big(\sin^2\!\!\!\; \epsilon\,\pro{0}{0}+\cos^2\!\!\!\; \epsilon\,\pro{1}{1}\big)+\frac{\nu}{2} \hat{I}.
\]
The quantum Fisher information for estimating $\epsilon$ from this state is
\[\label{eq:discontinuous_example_mixed_H}
H(\R_{\epsilon,\nu})\equiv H^{11}(\epsilon,\nu)=\frac{4(1-\nu)^2\sin^2(2\epsilon)}{1-(1-\nu)^2\cos^2(2\epsilon)}.
\]
(The notation was taken from Theorem~\ref{thm:regularization_procedure}, where $\epsilon$ is the only estimated parameter while $\nu$ is auxiliary at this point.) This function is depicted in Fig.~\ref{fig:reg}. It is undefined at points $(\epsilon,\nu)=(k\frac{\pi}{2},0)$, $k\in \mathbb{Z}$; however, as theorem~\ref{thm:regularization_procedure} shows, for $\nu\rightarrow 0$ this function must converge to the quantum Fisher information given by Eq.~\eqref{eq:disc_H_first_ex}.

Now we illustrate how the approximation~\eqref{eq:approx_H} performs compared to the exact value \eqref{eq:discontinuous_example_mixed_H}. We will derive the quantum Fisher information at point $(\epsilon,\nu)=(0.1,0.04)$. Inserting this point into Eq.~\eqref{eq:discontinuous_example_mixed_H} yields the exact value of the quantum Fisher information for the estimation of $\epsilon$,
\[\label{eq:exact_value}
H(\epsilon)\equiv H^{11}(0.1,0.04)=1.27.
\]

The extended unit vector for the approximation is defined by $\tilde{\be}\equiv (\epsilon,\nu)=(0.1,0.04)={(0,0)+\de\,\tilde{\boldsymbol{u}}}$, which gives $\tilde{\boldsymbol{u}}=\frac{(0.1,0.04)-(0,0)}{\norm{(0.1,0.04)-(0,0)}}$.
All eigenvalues of $\R_{\epsilon,\nu}$ are non-zero at $(\epsilon,\nu)=(0.1,0.04)$, which according to corollary~\ref{cor:equality_conditions_for_QFIs} means that the quantum Fisher information and the continuous quantum Fisher information are identical at this point. This in combination with Eq.~\eqref{eq:approx_H} yields
\[
H(\epsilon)=H_c(\epsilon)=H_c^{11}(\tilde{\be})\approx H_c^{11}(0,0)+\Delta_{\tilde{\boldsymbol{u}}}^{11}(0,0).
\]
Clearly $H_c^{11}(0,0)=4$. The only relevant eigenvalue needed for calculating $\Delta_{\tilde{\boldsymbol{u}}}^{11}(0,0)$ is $p_1(\epsilon,\nu)=(1-\nu)\sin^2\epsilon+\frac{\nu}{2}$, for which $p_1(0,0)=0$. The respective Hessian matrix is
\[
\mathcal{H}_1(0,0)=
\begin{bmatrix}
  2 & 0 \\
  0 & 0
\end{bmatrix},
\]
from which we calculate $\Delta_{\tilde{\boldsymbol{u}}}^{11}(0,0)=0$. This results in $H(\epsilon)\approx 4$ which is obviously not a good approximation.

To fix this problem we employ a simple trick: we substitute $\nu_1:=\sqrt{\nu}$ which reparametrizes Eq.~\eqref{eq:discontinuous_example_mixed_H} and effectively changes its graph, and then we perform the identical procedure at point ${(\epsilon,\nu_1)=(\epsilon,\sqrt{\nu})=(0.1,0.2)}$. The relevant Hessian matrix is now
\[
\mathcal{H}_1(0,0)=
\begin{bmatrix}
  2 & 0 \\
  0 & 1
\end{bmatrix},
\]
resulting in $\Delta_{\tilde{\boldsymbol{u}}}^{11}(0,0)=-\frac{8}{3}$. Finally, we have
\[
H(\epsilon)\approx H_c^{11}(0,0)+\Delta_{\tilde{\boldsymbol{u}}}^{11}(0,0)=4-\frac{8}{3}=1.33,
\]
which approximates the exact value very well.
\end{example}

\section{Discussion and conclusion}

We presented a theory that describes discontinuities of the quantum Fisher information matrix and we linked these discontinuities to the discrepancy between two figures of merit in quantum metrology, the quantum Fisher information matrix and the (four times) Bures metric. Although we have shown that the Bures metric represents in some sense a continuous version of the quantum Fisher information matrix, both the quantum Fisher information matrix and the Bures metric can be discontinuous in the topology of multiple parameters. These discontinuities and discrepancies appear at sets of measure zero and therefore can be often ignored. They also appear only when varying estimation parameters changes the rank of the density matrix describing a quantum state. Specifically, these problems never appear in the estimation of unitary channels using quantum probe states.

However, in certain scenarios these problems show up and can be a source of great confusion, as it it is not usual to see discontinuous functions in physics. These scenarios involve common tasks such as estimation of decoherence parameters, space-time parameters, temperature, or simultaneous estimation of multiple parameters. Moreover, it is often assumed that common tools such as the Taylor's expansion can be always employed, especially when assumptions are difficult to check, for instance when using a perturbative quantum field theory. But the Taylor's expansion of the quantum Fisher information cannot exist at points of discontinuity, and such points can appear even when the density matrix is an analytical function of its parameters.

It is not very clear how to interpret these discontinuities from a physical point of view. Expressions for the quantum Fisher information show that when the rank of the density matrix changes, there is a sudden drop in the precision with what we can estimate the parameter of interest. This can be connected to the fact that in such scenarios $\epsilon$ is not identifiable, i.e., $\epsilon$ and $-\epsilon$ produce exactly the same statistics of measurement results and therefore cannot be distinguished (remember example~\ref{ex:example1} in which $\R_{\epsilon}=\R_{-\epsilon}$). However, it is possible to design scenarios where $\epsilon$ can be identified by different means, for example by simultaneously observing the change in phase, but the drop in the quantum Fisher information still does not disappear. It is clear, however, that the sudden drop is always connected to the information that can be extracted from the change of purity, since it always depends on the derivatives of eigenvalues of the density matrix. This drop might also be a demonstration of a quantum phase transition which occurs at absolute zero. When temperature of a quantum probe goes to zero, the thermal state describing this probe suddenly becomes a pure state, resulting in a change of the rank of the density matrix, and consequently resulting in a discontinuity of the quantum Fisher information for the estimation of parameter of interest. A similar type of behavior has been reported in several papers~\cite{zanardi2007bures,gu2010fidelity,banchi2014quantum,wu2016geometric}. The physical meaning of these discontinuities, also in scenarios of bosonic systems, is discussed in more detail in Ref.~\cite{vsafranek2016gaussian}.

The quantum Cram\'er-Rao bound holds for the possibly discontinuous quantum Fisher information matrix. The last open question is whether such a bound can be derived for the Bures metric and under what circumstances is this possible. We leave answering these questions for future work.

\textit{\textbf{Acknowledgement}} I thank M\u{a}d\u{a}lin Gu\c{t}\u{a} and Animesh Datta for suggesting corrections to my thesis and valuable input. I thank Joshua M. Deutsch for listening carefully to my proof, Joseph C. Schindler for our discussion about induced metric, and Michael E2 for sharing his PlotPiecewise Wolfram Mathematica code.

\appendix

\section{Infinitesimal Bures distance for singular density matrices}\label{app}

In this appendix we are going to show that the logic used in paper~\cite{sommers2003bures} to derive the infinitesimal Bures distance $d_B^2(\R,\R+\dR)$ leads to a different result than the one stated in that paper. At the same time we argue that other papers based on this result are not usually affected in large part. This is because an expression similar to the one published in~\cite{sommers2003bures} can be obtained by considering $d_B^2(\R_{\be},\R_{\be+\bdeps})$ from Eq.~\eqref{eqn:bures} instead of $d_B^2(\R,\R+\dR)$. In other words, the right (expected) result comes from a different starting point. Moreover, when $\R$ is non-singular these two approaches lead to the same results.

To derive the correct expression for $d_B^2(\R,\R+\dR)$, we are going to redo the proof from~\cite{sommers2003bures} in detail using the same argumentation. By the same argumentation we mean that we will inherently assume that $\dR$ is the first-order correction to $\R$. This assumption also means that the starting point $d_B^2(\R,\R+\dR)$ cannot be used for the derivation of the Bures metric $g^{ij}$, because the proof we are going to show will be applicable only for $\dR$ linearly proportional to the estimation parameters, $\dR\sim\bdeps$, and \emph{will not} be applicable in case of  $\dR\sim\bdeps+\bdeps^2$. We will also use a notation similar to~\cite{sommers2003bures} so the reader can easily see the differences.

We are going to compute the Bures distance $d_B^2(\R_A,\R_B)$ between two infinitesimally close density matrices of size $N$. Let
us set $\R_A=\R$, $\R_B=\R+\dR$:
\[\label{eq:O_sommers}
\sqrt{\R_A^{1/2}\R_B\R_A^{1/2}}=\R+\hat{X}+\hat{Y},
\]
where the operator $\hat{X}$ is of order 1 in $\dR$, while $\hat{Y}$ is of order 2. Squaring this equation we obtain the first and the second
order
\[\label{eq:X_Y_0}
\R^{1/2}\dR\R^{1/2}=\hat{X}\R +\R\hat{X} ,\quad -\hat{X}^2 = \hat{Y}\R + \R\hat{Y}
\]
or in the basis $\{\ket{k}\}$ in which $\R$ is diagonal with eigenvalues $p_k$, $\R=\sum_kp_k\pro{k}{k}$,
for $p_k+p_l>0$ we obtain
\[\label{eq:X_Y_1}
\bra{k}\hat{X}\ket{l}=\bra{k}\dR\ket{l}\frac{p_k^{1/2}p_l^{1/2}}{p_k+p_l},\quad \bra{k}\hat{Y}\ket{l}=-\bra{k}\hat{X}^2\ket{l}\frac{1}{p_k+p_l}.
\]
We point out that the above equation also gives $\bra{k}\hat{X}\ket{l}=0$ and $\bra{k}\hat{Y}\ket{l}=0$ for either $p_k>0,p_l=0$ or $p_k=0,p_l>0$. (\emph{Now we start to differ from the proof in~\cite{sommers2003bures}.}) Because the left-hand side of Eq.~\eqref{eq:O_sommers} belongs to the subspace $\mathcal{L}(\HS_{>0})$, where $\HS_{>0}$ is the Hilbert space spanned by the eigenvectors associated with nonzero eigenvalue $p_k$, also the right hand side belongs to the same subspace. Hence for $p_kp_l=0$ we have
\[\label{eq:X_Y_2}
\bra{k}\hat{X}\ket{l}=0,\quad \bra{k}\hat{Y}\ket{l}=0.
\]
Combining Eqs.~\eqref{eq:X_Y_1} and \eqref{eq:X_Y_2} we obtain
\[\label{eq:X}
\hat{X}=\sum_{p_k>0,p_l> 0}\bra{k}\dR\ket{l}\frac{p_k^{1/2}p_l^{1/2}}{p_k+p_l}\pro{k}{l}.
\]
Applying the trace on this operator we find
\[\label{eq:help_traceX1}
\tr \hat{X}=\frac{1}{2}\tr \big[\hat{P}_{\HS_{>0}}\dR\big]=\frac{1}{2}\tr \big[(\hat{I}-{\hat{P}}_0)\dR\big]=-\frac{1}{2}\tr \big[\hat{P}_0\dR\big],
\]
where $\hat{P}_{\HS_{>0}}=\hat{I}-{\hat{P}}_0$ is the projector onto the previously mentioned subspace $\HS_{>0}$ and ${\hat{P}}_0$ is the projector onto the kernel of the density matrix $\R$. The last inequality is due to the fact that $\tr{\R} = 1$, and hence $\tr[\dR] = 0$.
Using the spectral decomposition of the density matrix $\R=\sum_kp_k\pro{k}{k}$ we have
\[\label{eq:help_traceX2}
\tr \big[\hat{P}_0\dR\big]\!=\!\tr \big[\!\hat{P}_0\big(\sum_k\mathrm{d}p_k\pro{k}{k}\!+\!p_k\pro{\mathrm{d}k}{k}\!+\!p_k\pro{k}{\mathrm{d}k}\big)\!\big]=0.
\]
The above expression vanishes because in every term
\begin{itemize}
  \item either $p_k=0$ and then also $\mathrm{d}p_k=0$ because $p_k$ achieves a local minimum,
  \item or $p_k>0$ and then by definition ${\hat{P}}_0\ket{k}=0$.
\end{itemize}
Combining Eqs.~\eqref{eq:help_traceX1} and~\eqref{eq:help_traceX2} we have
\[
\tr \hat{X}=0.
\]

Using Eqs.~\eqref{eq:X_Y_1}, \eqref{eq:X_Y_2}, and \eqref{eq:X} we obtain
\[
\tr\hat{Y}=-\sum_{p_k>0}\frac{\bra{k}\hat{X}^2\ket{k}}{2p_k}
=-\sum_{p_k>0,p_l>0}\frac{1}{4}\frac{|\bra{k}\dR\ket{l}|^2}{p_k+p_l}.
\]
Using definition~\eqref{def:bures_distance} we find $d_B^2(\R,\R+\dR)=-2(\tr\hat{X}+\tr\hat{Y})$ which gives the result for the infinitesimal Bures distance
\[\label{eq:inccorect_Bures}
(ds)_B^2=d_B^2(\R,\R+\dR)=\frac{1}{2}\sum_{p_k>0,p_l>0}\frac{|\bra{k}\dR\ket{l}|^2}{p_k+p_l}.
\]

Note that in contrast to the result of~\cite{sommers2003bures} who arrived at $(ds)_B^2=d_B^2(\R,\R+\dR)=\frac{1}{2}\sum_{p_k+p_l>0}\frac{|\bra{k}\dR\ket{l}|^2}{p_k+p_l}$, by applying the same argumentation in detail we arrived at the result where mixed terms for which either $p_k>0,p_l=0$ or $p_k=0,p_l>0$ are missing. Such a result does not seem to give the result for the metric tensor $g^{ij}$ that we would expect to obtain: we would expect something similar or identical to the quantum Fisher information matrix~\eqref{QFI_multi}. In fact, plainly inserting $\dR=\sum_i\partial_i\R\de_i$ into Eq.~\eqref{eq:inccorect_Bures} yields
\[\label{g_incorrect}
g^{ij}(\be)=\frac{1}{2}\sum_{p_k>0,p_l> 0}\frac{\Re(\langle k|\partial_i\R_{\be}|l\rangle\langle l|\partial_j\R_{\be}|k\rangle)}{p_k+p_l},
\]
which is not compatible with the quantum Fisher information matrix~\eqref{QFI_multi} as long as $\R$ is singular. However, we are getting this unexpected result because we misused the inherent assumption that $\dR\sim\bdeps$ is the first order correction, while the defining relation for the Bures metric~\eqref{eqn:bures} is defined by the second order terms in $\bdeps$. A simple fix by expanding $\dR$ into the second order, $\dR=\sum_i\partial_i\R\de_i+\frac{1}{2}\sum_{ij}\partial_{ij}\R\de_i\de_j$, also does not work and leads to the same result~\eqref{g_incorrect}. This result for $g^{ij}$ is, however, incorrect. This is simply because the proof in which we derived an expression for $d_B^2(\R,\R+\dR)$ is not applicable for deriving the expression for the Bures metric $g^{ij}$. This is because if we wanted to expand $\R_B$ up to the second order in $\bdeps$ there would be extra terms in the second expression in Eq.~\eqref{eq:X_Y_0}. These extra terms can be found in an equivalent of this equation, in Eq.~\eqref{eq:discrepancy_between_proofs}. Considering these extra terms then lead to the correct expression for the Bures metric as given by theorem~\ref{thm:theorem1},
\[
\begin{split}
g^{ij}(\be)&=\tfrac{1}{4}H^{ij}(\be)+\tfrac{1}{2}\!\!\!\!\sum_{p_k(\be)=0}\!\!\!\!\partial_{ij}p_k(\be)\\
&=
\tfrac{1}{2}\!\!\!\!\!\sum_{p_k+p_l> 0}\!\!\!\!\!\frac{\Re(\langle k|\partial_i\R_{\be}|l\rangle\langle l|\partial_j\R_{\be}|k\rangle)}{p_k+p_l}+\tfrac{1}{2}\!\!\!\!\sum_{p_k(\be)=0}\!\!\!\!\partial_{ij}p_k(\be).
\end{split}
\]

\section{Proof of theorem~\ref{thm:theorem1}}\label{app:theorem1}
\begin{proof}
We generalize proof from Ref.~\cite{hubner1992explicit} to include singular density matrices. By combining the defining relation for the Bures metric~\eqref{eqn:bures}, the definition of the Bures distance~\eqref{def:bures_distance}, and the definition of the Uhlmann fidelity~\eqref{def:Uhlmann_Fidelity} we obtain an expression for the Bures metric,
\[\label{eq:bures_metric_fid}
\sum_{i,j}g^{ij}(\be)\mathrm{d}\epsilon_i\mathrm{d}\epsilon_j=2\Big(1-\tr\Big[\sqrt{\sqrt{\R_{\be}}\R_{\be+\bdeps}\sqrt{\R_{\be}}}\Big]\Big).
\]

We define operator $\hat{O}(\bdeps):=\sqrt{\R_{\be}}\R_{\be+\bdeps}\sqrt{\R_{\be}}$. Because operator $\hat{O}(\bdeps)$ is given by applying $\sqrt{\R_{\be}}=\sum_k\sqrt{p_k}\pro{k}{k}$ on both sides of $\R_{\be+\bdeps}$, it can be written as
\[
\hat{O}(\bdeps)=\sum_{p_k>0,\,p_l>0}o_{kl}(\bdeps)\pro{k}{l}.
\]
(We omitted writing the explicit dependence on $\be$.) As a result, this operator clearly belongs to the subspace of linear operators acting on the Hilbert space spanned by eigenvectors associated with non-zero eigenvalues $p_k$, i.e., $\hat{O}\in\mathcal{L}(\HS_{>0})$, where $\HS_{>0}:=\mathrm{span}\{\ket{k}\}_{p_k>0}$. Now we define its square root $\hat{A}(\bdeps)$,
\[\label{eq:A_definition}
\hat{A}(\bdeps)\hat{A}(\bdeps)=\hat{O}(\bdeps).
\]
Because $\hat{O}\in\mathcal{L}(\HS_{>0})$, also $\hat{A}\in\mathcal{L}(\HS_{>0})$ together with all of its derivatives. To show that, we assume that $\hat{O}$ has a spectral decomposition\footnote{Spectral decomposition exists, because $\hat{O}$ is a Hermitian operator.} $\hat{O}(\bdeps)=\sum_mo_m^{\mathrm{diag}}(\bdeps) P_m(\bdeps)$, where $P_m(\bdeps)=\sum_{p_k>0,\,p_l>0}c^{(m)}_{kl}(\bdeps)\pro{k}{l}$. Such an expression is valid because operator $\hat{O}$ lies in the previously mentioned subspace $\mathcal{L}(\HS_{>0})$. The square root is then given by
\[
\begin{split}
\hat{A}(\bdeps)&=\sum_m\sqrt{o_m^{\mathrm{diag}}(\bdeps)} P_m(\bdeps)\\
&=\sum_{p_k>0,\,p_l>0}\left(\sum_{m}\sqrt{o_m^{\mathrm{diag}}(\bdeps)} c^{(m)}_{kl}(\bdeps)\right)\pro{k}{l}.
\end{split}
\]
Clearly, any derivatives of $\hat{A}(\bdeps)$ with respect to $\de_i$ will change only the factors, so the resulting operator will still remain in the same subspace $\mathcal{L}(\HS_{>0})$.

From Eq.~\eqref{eq:bures_metric_fid} we obtain
\[
\sum_{i,j}g^{ij}(\be)\mathrm{d}\epsilon_i\mathrm{d}\epsilon_j=2(1-\tr[A(\bdeps)]),
\]
which gives an expression for elements of the Bures metric,
\[\label{eq:g_with_A}
g^{ij}(\be)=-\tr[\partial_{\de_i\de_j}\!A(0)],
\]
if the second derivatives exist. For that reason we assume that $\R\in C^{(2)}$, i.e., the second derivatives of $\R$ exist and are continuous.\footnote{Actually, this assumption can be slightly weakened. We can assume that the second derivatives exists, but may not be necessarily continuous. But the continuity of the second derivatives implies $\partial_{ij}\R=\partial_{ji}\R$ which will be useful later in theorem~\ref{thm:theorem2} when discussing the continuity of the Bures metric.} To obtain these second partial derivatives we rewrite Eq.~\eqref{eq:A_definition} while expanding $\R_{\be+\bdeps}$ around point $\be$,
\[\label{eq:A_definition2}
\hat{A}(\bdeps)\hat{A}(\bdeps)=\sqrt{\R}\bigg(\R+\sum_k\partial_k\R\ \de_k + \frac{1}{2} \sum_{k,l}\partial_{kl}\R\ \de_k \de_l\bigg)\sqrt{\R}.
\]
By differentiating this equation with respect to $\de_i$  and setting $\bdeps=0$ we obtain
\[
\partial_{\de_i}\!\hat{A}(0)\,\R+\R\,\partial_{\de_i}\!\hat{A}(0)=\sqrt{\R}\partial_i\R\sqrt{\R},
\]
where we used $\hat{A}(0)=\R$. By applying $\bra{k}\ \ket{l}$ for $p_k>0$ and $p_l>0$ we obtain the matrix elements of $\partial_{\de_i}\!\hat{A}(0)$,
\[\label{eq:elem_of_Adot}
\bra{k}\partial_{\de_i}\!\hat{A}(0)\ket{l}=\frac{\sqrt{p_kp_l}\bra{k}\partial_i\R\ket{l}}{p_k+p_l}.
\]
Elements $\bra{k}\partial_{\de_i}\!\hat{A}(0)\ket{l}$ such that $p_k=0$ or $p_l=0$ are identically zero, because as we proved earlier all derivatives of $\hat{A}$ lie in the subspace $\mathcal{L}(\HS_{>0})$. Differentiating Eq.~\eqref{eq:A_definition2} for the second time and setting $\bdeps=0$ yields\footnote{This is where our proof starts to effectively differ from finding the expression for $d_B^2(\R,\R+\dR)$. In both \cite{hubner1992explicit,sommers2003bures} the right hand side of Eq.~\eqref{eq:discrepancy_between_proofs} is set to zero, as a result of a not explicitly stated (and easy to miss) assumption that $\dR$ can be linear only in $\bdeps$, $\dR\sim\bdeps$.}
\begin{multline}\label{eq:discrepancy_between_proofs}
    \partial_{\de_i\de_j}\!\hat{A}(0)\, \R+\{\partial_{\de_i}\!\hat{A}(0),\partial_{\de_j}\!\hat{A}(0)\}+\R\, \partial_{\de_i\de_j}\!\hat{A}(0)\\
=\sqrt{\R}\partial_{ij}\R\sqrt{\R},
\end{multline}
where $\{\,,\}$ denotes an anti-commutator. Now, restricting ourselves to the subspace $\mathcal{L}(\HS_{>0})$, the density matrix has the inverse matrix $\R^{-1}$ in this subspace. We multiply the above equation by this matrix and perform the trace on this subspace,
\begin{multline}\label{eq:second_derivatives_and_rho}
\tr_{\mathcal{L}(\HS_{>0})}[\rho^{-1}\{\partial_{\de_i}\hat{A}(0),\partial_{\de_j}\hat{A}(0)\}]\\
+2\tr_{\mathcal{L}(\HS_{>0})}[\partial_{\de_i\de_j}\hat{A}(0)]
=\tr_{\mathcal{L}(\HS_{>0})}[\partial_{ij}\R].
\end{multline}
Because all derivatives of $\hat{A}$ lie in the subspace $\mathcal{L}(\HS_{>0})$, traces of such operators are identical on both the subspace and the full space, $\tr_{\mathcal{L}(\HS_{>0})}[\partial_{\de_i\de_j}\hat{A}(0)]=\tr[\partial_{\de_i\de_j}\hat{A}(0)]$. However that is not necessarily true for the last element for which $\tr_{\mathcal{L}(\HS_{>0})}[\partial_{ij}\R]=\tr[\hat{P}_{\HS_{>0}}\partial_{ij}\R]$, where $\hat{P}_{\HS_{>0}}$ denotes the projector on the Hilbert space $\HS_{>0}$. Because $\tr[\partial_{ij}\R]=0$, this term can be equivalently written as $\tr[\hat{P}_{\HS_{>0}}\partial_{ij}\R]=-\tr[\hat{P}_0\partial_{ij}\R]$, where the projector $\hat{P}_0:=\hat{I}-\hat{P}_{\HS_{>0}}$ projects onto the subspace spanned by the eigenvectors of the density matrix $\R$ associated with the zero eigenvalue.
Therefore combining  Eqs.~\eqref{eq:g_with_A} and \eqref{eq:second_derivatives_and_rho}  yields
\[\label{eq:g_in_terms_of_A}
g^{ij}
=\frac{1}{2}\big(\tr_{\mathcal{L}(\HS_{>0})}[\rho^{-1}\{\partial_{\de_i}\!\hat{A}(0),\partial_{\de_j}\!\hat{A}(0)\}]
+\tr[\hat{P}_0\partial_{ij}\R]\big).
\]
The first term of the right hand side can be readily computed from Eq.~\eqref{eq:elem_of_Adot} while the antisymmetric part vanishes under the sum,
\begin{multline}\label{eq:first_term}
\tr_{\mathcal{L}(\HS_{>0})}[\rho^{-1}\{\partial_{\de_i}\!\hat{A}(0),\partial_{\de_j}\!\hat{A}(0)\}]=\\
\sum_{p_k>0,\,p_l>0}\!\!\!\!\frac{\Re(\bra{k}\partial_i\R\ket{l}\bra{l}\partial_j\R\ket{k})}{p_k+p_l}.
\end{multline}
Now we compute the second term. The second derivative of $\R$ is given by
\[\label{eq:second_derivative_rho}
\begin{split}
   \partial_{ij}&\R =\sum_{k} \partial_{ij}p_k\pro{k}{k}\\
   &+p_k\big(\pro{\partial_i k}{\partial_j k}+\pro{\partial_j k}{\partial_i k}\big)
   +p_k\big(\pro{\partial_{ij} k}{k} + \pro{k}{\partial_{ij} k}\big)\\
   &+\partial_jp_k\big(\pro{\partial_i k}{k}+\pro{k}{\partial_i k}\big)+\partial_ip_k\big(\pro{\partial_j k}{k}+\pro{k}{\partial_j k}\big).
\end{split}
\]
We stress out that the summation goes over all values of $k$, even over those for which $p_k=0$. When using the above equation to calculate $\tr[\hat{P}_0\partial_{ij}\R]$ we find that many terms vanish because:
\begin{itemize}
  \item for $k$ such that $p_k>0$, $\hat{P}_0\ket{k}=0$,
  \item for $k$ such that $p_k=0$, also $\partial_ip_k=\partial_jp_k=0$, because $p_k$ reaches the local minimum at point $\epsilon$ such that $p_k(\epsilon)=0$.
\end{itemize}
Only parts of the first two terms of Eq.~\eqref{eq:second_derivative_rho} remain,
\[\label{eq:second_term}
\begin{split}
\tr[\hat{P}_0\partial_{ij}\R]&=\sum_{p_k=0}\partial_{ij}p_k+2\!\!\!\!\!\!\sum_{p_k>0,p_l=0}\!\!\!\!\!p_k\ \Re(\braket{l}{\partial_i k}\braket{\partial_j k}{l})\\
&=\sum_{p_k=0}\partial_{ij}p_k+2\!\!\!\!\!\!\sum_{p_k>0,p_l=0}\!\!\!\!\!\!\frac{\Re(\bra{k}\partial_i\R\ket{l}\bra{l}\partial_j\R\ket{k})}{p_k+p_l}.
\end{split}
\]

Combining Eqs.~\eqref{eq:g_in_terms_of_A}, \eqref{eq:first_term}, and \eqref{eq:second_term}, and the expression for the quantum Fisher information matrix~\eqref{QFI_multi} we derive
\[\label{eq:Bures_metric_app}
\begin{split}
g^{ij}&=\frac{1}{2}\sum_{p_k+p_l> 0}\!\!\!\!\!\!\frac{\Re(\bra{k}\partial_i\R\ket{l}\bra{l}\partial_j\R\ket{k})}{p_k+p_l}
+\frac{1}{2}\sum_{p_k=0}\!\!\partial_{ij}p_k\\
&=\frac{1}{4}\Big(H^{ij}+2\!\!\!\sum_{\,p_k=0}\!\!\!\partial_{ij}p_k\Big),
\end{split}
\]
which considering the definition of $H_c$, Eq.~\eqref{def:continuous_QFI}, proves the theorem.
\end{proof}

\section{Proof of theorem~\ref{thm:theorem2}}\label{app:theorem2}
\begin{proof}
First we will study the neighborhood of the quantum Fisher information matrix, i.e., we will study the function $H^{ij}(\be+\bdeps)$, where we define $\bdeps=(\de_1,\dots,\de_n)$. Then we show that the (four times) Bures metric $H_c$ is given by the limits stated in the theorem. Finally we prove the continuity property stated in the theorem.

Equation~\eqref{QFI_multi} gives
\begin{equation}\label{eq:QFI_in_the_neighborhood}
\begin{split}
&H^{ij}(\be+\bdeps) \\
&=2\!\!\!\!\!\!\!\!\!\!\sum_{p_{k\,\be+\bdeps}+p_{l\,\be+\bdeps}> 0}\!\!\!\!\!\!\!\!\!\!\frac{\Re(\bra{k_{\be+\bdeps}}\partial_i\R_{\be+\bdeps}\ket{l_{\be+\bdeps}}\bra{l_{\be+\bdeps}}\partial_j\R_{\be+\bdeps}\ket{k_{\be+\bdeps}})}{p_{k\,\be+\bdeps}+p_{l\,\be+\bdeps}} \\
&=2\!\!\!\!\!\!\sum_{p_{k\,\be}+p_{l\,\be}> 0}\!\!\!\!\!\!
\frac{\Re(\bra{k_{\be}}\partial_i\R_{\be}\ket{l_{\be}}\bra{l_{\be}}\partial_j\R_{\be}\ket{k_{\be}}+\mathcal{O}(\bdeps))}{p_{k\,\be}+p_{l\,\be}+\mathcal{O}(\bdeps)} \\
&+2\!\!\!\!\!\!\!\!\!\sum_{\substack{p_{k\,\be}+p_{l\,\be}= 0, \\{p_{k\,\be+\bdeps}+p_{l\,\be+\bdeps}> 0}}}\!\!\!\!\!\!\!\!\!\frac{\Re(\bra{k_{\be+\bdeps}}\partial_i\R_{\be+\bdeps}\ket{l_{\be+\bdeps}}\bra{l_{\be+\bdeps}}\partial_j\R_{\be+\bdeps}\ket{k_{\be+\bdeps}})}{p_{k\,\be+\bdeps}+p_{l\,\be+\bdeps}}.
\end{split}
\end{equation}
$\mathcal{O}(\bdeps)$ (the big O notation) denotes the remainder after expanding $p_{k\,\be+\bdeps}+p_{l\,\be+\bdeps}$.  In the above equation we have chosen $\bdeps$ small enough such that $\abs{\mathcal{O}(\bdeps)}<p_{k\,\be}+p_{l\,\be}$ for $p_{k\,\be}+p_{l\,\be}> 0$.

Using $\R_{\be}\in C^{(2)}$ (which is assumed in all theorems in this paper) we can write the following expansion:
\begin{align}\label{eq:expansions_prhok}
p_{k\,\be+\bdeps}&=p_{k}\!+\!\sum_m\partial_mp_{k}\de_m\!+\!\frac{1}{2}\!\sum_{m,n}\partial_{mn}p_{k}\de_m\de_n+\mathcal{O}(\bdeps^3),\nonumber\\
\R_{\be+\bdeps}&=\R+\sum_m\partial_m\R\de_m+\frac{1}{2}\sum_{m,n}\partial_{mn}\R\de_m\de_n+\mathcal{O}(\bdeps^3),\nonumber\\
\ket{k_{\be+\bdeps}}&=\ket{k}\!+\!\sum_m\ket{\partial_mk}\de_m\!+\!\frac{1}{2}\!\sum_{m,n}\ket{\partial_{mn}k}\de_m\de_n\!+\!\mathcal{O}(\bdeps^3).
\end{align}
$\mathcal{O}(\bdeps^3)$ denotes the remainder that consists of sums of multiples of three or more elements of vector $\bdeps$. Using these expansions and Eq.~\eqref{eq:second_derivative_rho}, for $k$, $l$ such that $p_{k\,\be}+p_{l\,\be}= 0$ we have $\partial_ip_{k\,\be}+\partial_ip_{l\,\be}=0$ and
\begin{align}\label{eq:elements_p_and_kl}
&p_{k\,\be+\bdeps}+p_{l\,\be+\bdeps}=\frac{1}{2}\sum_{m,n}(\partial_{mn}p_{k}+\partial_{mn}p_{l})\de_m\de_n\!+\!\mathcal{O}(\bdeps^3),\nonumber\\
&\bra{k_{\be+\bdeps}}\partial_i\R_{\be+\bdeps}\ket{l_{\be+\bdeps}}\nonumber\\
&\ \ = \sum_m\big(\bra{\partial_mk}\partial_i\R\ket{l}\!+\!\bra{k}\partial_i\R\ket{\partial_ml}\!+\!\bra{k}\partial_{im}\R\ket{l}\big)\de_m\!+\!\mathcal{O}(\bdeps^2)\nonumber\\
&\ \ =\delta_{kl}\sum_m\partial_{im}p_k\de_m+\mathcal{O}(\bdeps^2),
\end{align}
where we used $\braket{k}{\partial_i j}=-\braket{\partial_i k}{j}$ which comes from the orthonormality condition. Inserting Eqs.~\eqref{eq:expansions_prhok} for $p_k+p_l>0$ and Eqs.~\eqref{eq:elements_p_and_kl} for $p_k+p_l=0$ into Eq.~\eqref{eq:QFI_in_the_neighborhood} yields
\[\label{eq:H_expanded}
\begin{split}
H^{ij}(\be&+\bdeps)=2\!\!\!\!\!\!\sum_{p_{k}+p_{l}> 0}\!\!\!\!\!\!
     \frac{\Re(\bra{k}\partial_i\R\ket{l}\bra{l}\partial_j\R\ket{k})+\mathcal{O}(\bdeps)}{p_{k}+p_{l}+\mathcal{O}(\bdeps)}\\ &+2\!\!\!\!\!\!\!\!\!\!\!\!\!\!\!\!\sum_{\substack{p_{k}= 0, \\\sum_{s,t}\partial_{st}p_{k}\de_s\de_t> 0}}\!\!\!\!\!\!\!\!\!\!\!\!\!\!\!\!\frac{\big(\sum_{m}\partial_{im}p_k \de_m\big) \big(\sum_{n}\partial_{jn}p_k\de_n\big)+\mathcal{O}(\bdeps^3)}{\sum_{s,t}\partial_{st}p_{k}\de_s\de_t+\mathcal{O}(\bdeps^3)}.
\end{split}
\]
By setting $\bdeps=\de\boldsymbol{e}_i$ and performing the limit we find
\[\label{eq:proof_of_thm:theorem2}
\begin{split}
\lim_{\de\rightarrow 0}H^{ij}(\be+\de \boldsymbol{e}_i)&=H^{ij}(\be)+2\!\!\!\!\sum_{\substack{p_{k}(\be)= 0, \\\partial_{ii}p_{k}(\be)> 0}}\!\!\!\!\partial_{ij}p_k(\be)\\
&=H^{ij}(\be)+2\!\!\sum_{p_{k}(\be)= 0}\!\!\partial_{ij}p_k(\be),
\end{split}
\]
which is equal to $H_c^{ij}(\be)$ according to theorem~\ref{thm:theorem1}. The second equality in Eq.~\eqref{eq:proof_of_thm:theorem2} is due to the Sylvester's criterion for positive semi-definite matrices~\cite{horn2012matrix} which gives $\partial_{ii}p_k(\be)\partial_{jj}p_k(\be)-\partial_{ij}p_k(\be)^2\geq0$, i.e., for $p_k(\be)=0$ and $\partial_{ii}p_k(\be)=0$ also $\partial_{ij}p_k(\be)=0$.
 The same equality holds for $\bdeps=\de \boldsymbol{e}_j$ which proves the first part of the theorem.

Now we are going to prove the continuity property stated in the theorem. We set $\bdeps=\de \boldsymbol{u}$, where $\boldsymbol{u}=(u_1,\dots,u_n)$ is a unit vector. Using
\[
\begin{split}
\frac{x+\mathcal{O}(\de)}{y+\mathcal{O}(\de)}&=\frac{x}{y\big(1+\frac{\mathcal{O}(\de)}{y}\big)}+\frac{\mathcal{O}(\de)}{y+\mathcal{O}(\de)}\\
&=\frac{x}{y}\bigg(1-\frac{\mathcal{O}(\de)}{y}\bigg)+\frac{\mathcal{O}(\de)}{y+\mathcal{O}(\de)}=\frac{x}{y}+\mathcal{O}(\de),
\end{split}
\]
which holds for any $y\neq 0$, and Eq.~\eqref{eq:H_expanded} while assuming the number of eigenvalues $p_k$ is finite we derive
\begin{multline}\label{eq:H_expanded2}
H^{ij}(\be+\de\boldsymbol{u})=\\
H^{ij}(\be)+\!2\!\!\!\!\!\!\!\!\!\!\!\!\!\!\!\!\sum_{\substack{p_{k}= 0, \\\sum_{s,t}\partial_{st}p_{k}u_su_t> 0}}
\!\!\!\!\!\!\!\!\!\!\!\!\!\!\!\!\frac{\big(\sum_{m}\partial_{im}p_k u_m\big) \big(\sum_{n}\partial_{jn}p_ku_n\big)}{\sum_{s,t}\partial_{st}p_{k}u_su_t}+\mathcal{O}(\de).
\end{multline}
By definition, the function $H_c^{ij}$ is continuous in $\epsilon_i$ at point $\be$ when
\[
(\forall \gamma>0)(\exists \delta>0)(\forall \de, \abs{\de}<\delta)(\abs{H_c^{ij}(\be+\de\boldsymbol{e}_i)-H_c^{ij}(\be)}<\gamma).
\]
Setting $\boldsymbol{u}:=\boldsymbol{e}_i$, using theorem~\ref{thm:theorem1} and Eq.~\eqref{eq:H_expanded2} we derive
\[\label{eq:long_proof_of_limit}
\begin{split}
&\abs{H_c^{ij}(\be+\de\boldsymbol{e}_i)-H_c^{ij}(\be)}
=\abs{H^{ij}(\be+\de\boldsymbol{e}_i)-H^{ij}(\be)\\
&+2\!\!\!\!\!\!\!\!\!\!\!\!\sum_{p_k(\be+\de\boldsymbol{e}_i)=0}\!\!\!\!\!\!\!\!\!\!\!\!\partial_{ij}p_k(\be+\de\boldsymbol{e}_i)
-2\!\!\!\!\sum_{p_k(\be)=0}\!\!\!\!\partial_{ij}p_k(\be)}\\
&=\Big|\ 2\!\!\!\!\!\!\!\sum_{\substack{p_{k}(\be)= 0, \\\partial_{ii}p_k(\be)> 0}}\!\!\!\!\!\frac{\partial_{ii}p_k(\be) \partial_{ji}p_k(\be) \de^2}{\partial_{ii}p_k(\be) \de^2}+\mathcal{O}(\de)-2\!\!\!\!\sum_{p_k(\be)=0}\!\!\!\!\partial_{ij}p_k(\be)\\
&+2\!\!\!\!\!\!\!\!\!\!\!\!\sum_{p_k(\be+\de\boldsymbol{e}_i)=0}\!\!\!\!\!\!\!\!\!\!\!\!\big(\partial_{ij}p_k(\be+\de\boldsymbol{e}_i)-\partial_{ij}p_k(\be)\big)
+2\!\!\!\!\!\!\!\!\!\!\!\!\sum_{p_k(\be+\de\boldsymbol{e}_i)=0}\!\!\!\!\!\!\!\!\!\!\!\!\partial_{ij}p_k(\be)\Big|\\
&\leq2\!\!\!\!\!\!\!\!\!\!\!\!\sum_{p_k(\be+\de\boldsymbol{e}_i)=0}\!\!\!\!\!\!\!\!\!\!\!\!|\partial_{ij}p_k(\be+\de\boldsymbol{e}_i)-\partial_{ij}p_k(\be)|\\
&+2\!\!\!\!\!\!\!\!\!\!\!\!\sum_{p_k(\be+\de\boldsymbol{e}_i)=0}\!\!\!\!\!\!\!\!\!\!\!\!|\partial_{ij}p_k(\be)|+|\mathcal{O}(\de)|<\frac{\gamma}{3}+\frac{\gamma}{3}+\frac{\gamma}{3}=\gamma.
\end{split}
\]
The first inequality is the triangle inequality. $\partial_{ij}p_k(\be)=\partial_{ji}p_k(\be)$ follows from $\R_{\be}\in C^{(2)}$ ($p_k\in C^{(2)}$, i.e., the second derivative is continuous). The same property also implies $2\sum_{p_k(\be+\de\boldsymbol{e}_i)=0}|\partial_{ij}p_k(\be+\de\boldsymbol{e}_i)-\partial_{ij}p_k(\be)|<\frac{\gamma}{3}$ for small enough $\de$ (i.e., for all $\de$ such that $|\de|<\delta$ where $\delta>0$ is some sufficiently small radius). $2\sum_{p_k(\be+\de\boldsymbol{e}_i)=0}|\partial_{ij}p_k(\be)|<\frac{\gamma}{3}$ holds for small enough $\de$ because
\begin{itemize}
  \item either $p_k(\be)=0$, and then from $0=p_k(\be+\de\boldsymbol{e}_i)=p_k(\be)+\frac{1}{2}\partial_{ii}p_k(\be)\de^2=\frac{1}{2}\partial_{ii}p_k(\be)\de^2$ follows $\partial_{ii}p_k(\be)=0$. Using the Sylvester's criterion again we have $\partial_{ij}p_k(\be)=0$, i.e., the corresponding term in the sum is zero.
  \item or $p_k(\be)>0$, and then the continuity of $p_k$ implies that for small enough $\de$ also $p_k(\be+\de\boldsymbol{e}_i)>0$, i.e., the corresponding term does not appear in the sum.
\end{itemize}
At last,
$|\mathcal{O}(\de)|<\frac{\gamma}{3}$ comes from the definition of $\mathcal{O}(\de)$, which can be made arbitrarily small, i.e., we can choose $\delta$ such that for all $\de$, $\abs{\de}<\delta$, $|\mathcal{O}(\de)|<\frac{\gamma}{3}$, which proves the theorem.
\end{proof}

\section{Proof of theorem~\ref{thm:discontinuities}}\label{app:discontinuities}
\begin{proof}
We are going to prove Eq.~\eqref{eq:meaning_of_Delta} first. Other statements of the theorem will follow easily. To do that we generalize the second part of the proof of theorem~\ref{thm:theorem2}. Combining theorem~\ref{thm:theorem1}, Eq.~\eqref{eq:H_expanded2}, and definition~\eqref{eq:Delta} yields
\[\label{eq:set_of_eq_and_ineq1}
\begin{split}
&|H_c^{ij}(\be+\de\,\bu)-H_c^{ij}(\be)-\Delta_{\boldsymbol{u}}^{ij}(\be)|\\
&=\big|\ 2\!\!\!\!\!\!\!\!\!\!\!\!\sum_{p_k(\be+\de\bu)=0}\!\!\!\!\!\!\!\!\!\!\!\!\partial_{ij}p_k(\be+\de\bu)
-\!2\!\!\!\!\!\!\!\!\!\!\!\!\!\!\!\!\sum_{\substack{p_{k}(\be)= 0, \\\sum_{s,t}\partial_{st}p_{k}(\be)u_su_t= 0}}
\!\!\!\!\!\!\!\!\!\!\!\!\!\!\!\!\partial_{ij}p_k(\be)
+\mathcal{O}(\de)\big|\\
&\leq\ 2\!\!\!\!\!\!\!\!\!\!\!\!\sum_{p_k(\be+\de\bu)=0}\!\!\!\!\!\!\!\!\!\!\!\!|\partial_{ij}p_k(\be+\de\bu)-\partial_{ij}p_k(\be)|\\
&+2\ \big|\ \!\!\!\!\!\!\!\!\!\!\!\!\sum_{p_k(\be+\de\bu)=0}\!\!\!\!\!\!\!\!\!\!\!\!\partial_{ij}p_k(\be)-\!\!\!\!\!\!\!\!\!\!\!\!\!\!\!\!\!\sum_{\substack{p_{k}(\be)= 0, \\\sum_{s,t}\partial_{st}p_{k}(\be)u_su_t= 0}}
\!\!\!\!\!\!\!\!\!\!\!\!\!\!\!\!\partial_{ij}p_k(\be)\big|
+\big|\mathcal{O}(\de)\big|\\
&<\frac{\gamma}{3}+\frac{\gamma}{3}+\frac{\gamma}{3}=\gamma,
\end{split}
\]
which proves Eq.~\eqref{eq:meaning_of_Delta}. $2\sum_{p_k(\be+\de\bu)=0}|\partial_{ij}p_k(\be+\de\bu)-\partial_{ij}p_k(\be)|<\frac{\gamma}{3}$ in Eq.~\eqref{eq:set_of_eq_and_ineq1} comes from the continuity of second derivatives. $2\big|\sum_{p_k(\be+\de\bu)=0}\partial_{ij}p_k(\be)-\sum_{\substack{p_{k}(\be)= 0, \\\sum_{s,t}\partial_{st}p_{k}(\be)u_su_t= 0}}
\partial_{ij}p_k(\be)\big|<\frac{\gamma}{3}$ because
\begin{itemize}
  \item either $p_k(\be)=0$ and then from $0=p_k(\be+\de\bu)=p_k(\be)+\frac{1}{2}\sum_{s,t}\partial_{st}p_{k}(\be)u_su_t\de^2$ follows $\sum_{s,t}\partial_{st}p_{k}(\be)u_su_t=0$. In this case the element $\partial_{ij}p_k(\be)$ in the first sum $\sum_{p_k(\be+\de\bu)=0}\partial_{ij}p_k(\be)$ is compensated by the element $\partial_{ij}p_k(\be)$ in the second sum.
  \item or $p_k(\be)>0$ and then the continuity of $p_k$ implies that for small enough $\de$ also $p_k(\be+\de\bu)>0$, i.e., for small enough $\de$ term $\partial_{ij}p_k(\be)$ does not appear in the first sum. The corresponding term also does not appear in the second sum because only terms for which $p_k(\be)=0$ are counted.
\end{itemize}
In total we have
\[
\begin{split}
&~ 2\ \big|\ \!\!\!\!\!\!\!\!\!\!\!\!\sum_{p_k(\be+\de\bu)=0}\!\!\!\!\!\!\!\!\!\!\!\!\partial_{ij}p_k(\be)-\!\!\!\!\!\!\!\!\!\!\!\!\!\!\!\!\!\sum_{\substack{p_{k}(\be)= 0, \\\sum_{s,t}\partial_{st}p_{k}(\be)u_su_t= 0}}
\!\!\!\!\!\!\!\!\!\!\!\!\!\!\!\!\partial_{ij}p_k(\be)\big|\\
&\leq 2\ \big|\ \!\!\!\!\!\!\!\!\!\!\!\!\sum_{\substack{p_{k}(\be)= 0, \\{p_k(\be+\de\bu)=0}}}\!\!\!\!\!\!\!\!\!\!\!\!\partial_{ij}p_k(\be)-\!\!\!\!\!\!\!\!\!\!\!\!\!\!\!\sum_{\substack{p_{k}(\be)= 0, \\\sum_{s,t}\partial_{st}p_{k}(\be)u_su_t= 0}}
\!\!\!\!\!\!\!\!\!\!\!\!\!\!\partial_{ij}p_k(\be)\big|
+2\ \big|\ \!\!\!\!\!\!\!\!\!\!\!\!\sum_{\substack{p_{k}(\be)> 0, \\{p_k(\be+\de\bu)=0}}}\!\!\!\!\!\!\!\!\!\!\!\!\partial_{ij}p_k(\be)\big|\\
&\leq 0+\frac{\gamma}{3}.
\end{split}
\]
$|\mathcal{O}(\de)|<\frac{\gamma}{3}$ comes from the definition of $\mathcal{O}(\de)$, which can be made arbitrarily small.

By definition, $H_c$ is continuous in $\epsilon_l$ if and only if $\lim_{\de\rightarrow 0}H_c^{ij}(\be+\de\,\boldsymbol{e}_l)-H_c^{ij}(\be)=0$, which together with Eq.~\eqref{eq:meaning_of_Delta} proves the second part of the theorem.

The function $H_c$ is continuous at point $\be$ if and only if
\[
(\forall \gamma>0)(\exists \delta>0)(\forall \bdeps, \norm{\bdeps}<\delta)(\abs{H_c^{ij}(\be+\bdeps)-H_c^{ij}(\be)}<\gamma).
\]
By negating this statement we obtain
\[\label{eq:negation}
(\exists \gamma>0)(\forall \delta>0)(\exists \bdeps, \norm{\bdeps}<\delta)(\abs{H_c^{ij}(\be+\bdeps)-H_c^{ij}(\be)}\geq\gamma).
\]
Equation~\eqref{eq:meaning_of_Delta} yields
\[
|H_c^{ij}(\be+\de\,\bu)-H_c^{ij}(\be)-\Delta_{\boldsymbol{u}}^{ij}(\be)|<\frac{|\Delta_{\boldsymbol{u}}^{ij}(\be)|}{2},
\]
for all $\de$ such that $|\de|<\delta_1$, where $\delta_1>0$ is some sufficiently small radius.
To show the validity of Eq.~\eqref{eq:negation} we choose $\gamma:=\frac{|\Delta_{\boldsymbol{u}}^{ij}(\be)|}{2}$ and $\bdeps:=\de\bu$ with $\de$ any such that $|\de|<\min\{\delta,\delta_1\}$. Then
\[
\begin{split}
&\abs{H_c^{ij}(\be+\bdeps)-H_c^{ij}(\be)}\\
&\geq||\Delta_{\boldsymbol{u}}^{ij}(\be)|-|H_c^{ij}(\be+\de\,\bu)-H_c^{ij}(\be)-\Delta_{\boldsymbol{u}}^{ij}(\be)||\\
&\geq|\Delta_{\boldsymbol{u}}^{ij}(\be)|-\frac{|\Delta_{\boldsymbol{u}}^{ij}(\be)|}{2}=\frac{|\Delta_{\boldsymbol{u}}^{ij}(\be)|}{2}=\gamma,
\end{split}
\]
where the first inequality is a version $|a-b|\geq ||a|-|b||$ of the triangle inequality. This proves the first part of the theorem.
\end{proof}

\bibliographystyle{apsrev4-1}
\bibliography{Discontinuities_of_QFI_BiBTeX}

\end{document}